\documentclass[11pt]{amsart}

\usepackage{latexsym,amsxtra,amscd,ifthen}
\usepackage{amsfonts}
\usepackage{verbatim}
\usepackage{amsmath}
\usepackage{amsthm}
\usepackage{amssymb}
\usepackage{ytableau}

\numberwithin{equation}{section}

\theoremstyle{plain}
\newtheorem{theorem}{Theorem}[section]

\newtheorem{prop}[theorem]{Proposition}
\newtheorem{lemma}[theorem]{Lemma}
\newtheorem{cor}[theorem]{Corollary}

\newtheorem{conj}[theorem]{Conjecture}

\theoremstyle{definition}
\newtheorem{defi}[theorem]{Definition}

\newtheorem{exam}[theorem]{Example}

\newtheorem{rema}[theorem]{Remark}

\newcommand{\CC}{\ensuremath{\mathbb C}}
\newcommand{\EE}{\ensuremath{\mathcal{E}}}

\newcommand{\nfree}{\ensuremath{\mathcal{N}}} 

\newcommand{\ZZ}{\ensuremath{\mathbb Z}}

\newcommand{\coheight}{\mathop{\mathrm{coheight}}}
\def\Coker{\mbox{Coker\,}}

\def\deg{\mbox{deg\,}}
\def\dim{\ensuremath{\mbox{dim\,}}}
\def\ext{\mbox{Ext}}
\def\gh{\mbox{gh\,}}   
\def\gd{\mbox{gl. dim\,}}

\def\gr{\mbox{gr\,}}
\newcommand{\height}{\mathop{\mathrm{height}}}
\newcommand{\Hom}{\mathop{\mathrm{Hom}}}
    \newcommand{\homol}{\mathop{\mathrm{H\!}}}

\def\Id{\mbox{Id\,}}
\def\im{\mbox{Im\,}}  
\def\ker{\mbox{Ker\,}}

\def\lt{\mbox{lt\,}}
\def\mod{\mbox{mod\,}}

\def\ord{\mbox{ord\,}}

\newcommand{\pd}{\mathop{\mathrm{pd}}}

\def\Res{\mbox{Res\,}}
\def\rk{\mbox{rk\,}}

\newcommand{\supp}{\mathop{\mathrm{supp}}}
\def\tor{\mbox{Tor\,}}

\def\Zpl{\mbox{\bf Z}_+}

\begin{document}

\title{Degree of freedom count in linear gauge invariant PDE systems}

\author{Simon Lyakhovich and Dmitri Piontkovski}

     \address{Laboratory of Theoretical and Mathematical Physics, Physics Faculty, Tomsk State University, Tomsk, 634050, Russia \\ Department of Mathematics for Economics,
Myasnitskaya str. 20, HSE University, 
Moscow 101990, Russia}

\email{sll@phys.tsu.ru, dpiontkovski@hse.ru}



\date{\today}

 \begin{abstract}
Suppose a system of not necessarily Lagrangian partial linear differential equations (PDE) with constant coefficients describes a classical field theory. Einstein proposed a definition of the ``strength'' of such a field theory that defines its degree of freedom (DoF). Einsteinian strength is based on the asymptotic number of free Taylor series coefficients of bounded degree in the general solution of the PDE system. The direct count of DoF in this way is a complex and
technically demanding process. Here, proceeding from Einsteinian strength of equations and making use of commutative algebra tools, we deduce another DoF count recipe which is formulated in terms of orders of the field equations, their gauge symmetries and gauge identities. This DoF count recipe covers the case of reducible gauge field theories, and it is easy to use.

We begin with interpreting the matrix of the system as a linear map between polynomial modules. First, proceeding from Einsteinian definition, we derive an explicit formula for DoF as the multiplicity of a certain extension module. Second, we prove (for homogeneous and certain more general systems) another explicit formula for  DoF in terms of orders of equations and gauge generators. A notable consequence of this formula is that two Hermitian conjugate systems have identical DoF. 

Every system of classical field theory defines the BRST complex which has the natural ${\mathbb Z}$-grading known as the ghost number. We equip this complex with another ${\mathbb Z}$-grading, which we call differential order.  This grading is 1 for every space time-derivative, while all the fields, ghosts, and antifields are assigned with this degree in a certain way, depending on their ghost number and the orders of equations and gauge generators. We compute the Euler characteristic of the BRST complex with respect to this new ${\mathbb Z}$-grading. This provides homological interpretation of DoF for linear gauge system as minus
the residue at infinity of the logarithmic derivative of the Euler characteristics for its BRST complex.  
 \end{abstract}

\maketitle

\section{Introduction}
The main goal of this article is to provide a convenient and explicitly covariant tool for computing the physical degree of freedom (DoF) for classical linear field theories. The classical field equations are assumed to form a system of partial differential equations (PDE) with constant coefficients. 

In classical field theory, there are two conceptually different ways of defining DoF. One way is based on the idea of the number of initial data required to fix a unique solution for the PDE system modulo gauge transformations. This concept implies considering field equations as evolutionary ones, with time derivatives playing a special role, unlike derivatives with respect to the other independent variables, considered as ``space coordinates''. In this concept,  DoF is understood as the number of independent arbitrary gauge invariant functions on the space involved in the general solution of the system. 

For a PDE system arising from the least action principle, the equations can be reformulated in a constrained Hamiltonian form. Applying the Dirac-Bergmann algorithm to identify secondary constraints, one can always find the gauge symmetry and count DoF~\cite{Dirac}.
Non-variational differential equations can also be systematically brought into certain first-order normal forms with respect to time derivatives \cite{Lyakhovich:2008hu}. This method, being, in a sense, a non-variational extension of the Dirac-Bergmann constrained analysis, also allows one to find the gauge symmetry of the system and to count DoF. However, these DoF control methods are not manifestly covariant, which makes their application to relativistic field theories problematic.  
Furthermore, these methods imply bringing the PDE system
to certain normal forms with respect to time derivatives of the fields, which often may require inverting differential operators with respect to space coordinates. This may break the spatial locality and seems inappropriate for the local field theory.

The second approach to counting DoF for PDE is based on 
the notion of Einsteinian strength of equations
\cite{einstein}. We briefly explain this DoF count method in Subsection 2.2; for a systematic exposition and applications, see \cite{Seiler}. The general idea of this method is to evaluate the dimension of the solution space for the equations by studying the number
of the Taylor series coefficients of bounded degree in the general solution. 
The limit of this number divided by the same number
of the Taylor series coefficients for the unrestricted general function is then proportional to the degree of freedom. The advantage of the method is that it is a manifestly covariant procedure. 
It has been  
demonstrated under certain assumptions that the Einsteinian strength method leads to the same DoF number as the amount of initial data, see \cite[Ch.~8]{Seiler}, \cite{Schutz1975}. The strength method is sometimes used in specific field theory models, especially when the manifest covariance matters; see, for example, \cite{Garecki:2002xc}, \cite{Belenchia}. The limitation of this method, which restricts its application, lies in the fact that studying the asymptotic growth of a dimension is a complex and
technically demanding process.
 Even for the D'Alambert equation, this is not an immediate count. 

In \cite{KLS2013}, a manifestly covariant method for DoF counting  is proposed for general PDE systems. This method requires bringing the system into a certain form of  involution and determining sequences of gauge symmetries and gauge identities for the involutive closure.  These sequences, which are detailed in the next section, are used for an immediate DoF computation in terms of the orders of the differential operators corresponding to the equations, gauge symmetries, and gauge identities.
For three-term sequences (i.e., in the case of irreducible gauge symmetries and gauge identities), \cite{KLS2013} provides the proof that this method agrees with the Einsteinian DoF definition. For the general case of longer sequences, a conjecture  is proposed \cite{KLS2013} for DoF computation in terms of all gauge structure orders though this conjecture remains unproven. This conjectural formula is widely used in various models of gravity and higher spin fields (e.g., see \cite{Basile:2022mif},  \cite{Rahman:2020qal}, \cite{Joung:2016naf}, \cite{Kulaxizi:2014yxa}, \cite{AFHL2022}),
 because the manifest covariance matters for these models, while the Einsteinian strength count is too cumbersome in these theories.
 We clarify the meaning of the differential orders used in this  formula and prove a more general version of this conjectural formula  in the situations  where the orders are well-defined.   
 
Our first objective is to 
derive for linear PDEs an explicit formula for DoF in the Einsteinian sense using the language of commutative algebra (Theorem~\ref{th: deg_free_via_Ext}). 
That is, let $T$ denote the matrix defining the left-hand side of the system of PDEs. Since we consider the equations with constant coefficients, the entries of $T$ belong to the ring $P = \CC [\partial_0, \dots, \partial_{d-1}]$ of differential operators. Thus,  $T$ defines some map of free modules over the ring $P$.
Then DoF can be expressed in terms of the cokernel $W$ of this map as
$$
{\mathcal N} = 
    e(\ext^1_P(W,P), d-1).
$$
The formula may seem unusual, as using the commutative module multiplicity $e$ and the extension functor $\ext^1$ is uncommon in physics. Nevertheless, it provides a direct and efficient way to calculate DoF. For a given system in a specific dimension $d$, the computation requires only a few lines of code in a common computer algebra system (e.g.: Macaulay2). Proceeding from this reformulation of Einstein's DoF definition, we deduce the above simple DoF counting algorithm~\cite{KLS2013}. 

Moreover, we give a homological interpretation 
of the degree of freedom in terms of BRST complex of the PDE system. 
To give such an interpretation, 
we restrict ourselves to homogeneous systems, that is, we assume that one can assign  dimensions to the fields  in such a way that the equations of the system become homogeneous. 
Equivalently, we assign integer differential order to each field variable to ensure that in each  equation of the system, all terms have the same differential order (assuming that the order of  partial derivatives is one). Then the differential order is define for all generatings of the BRST complex. Moreover, the BRST differential has zero differential order. If the initial differential orders of the filed variables are sufficiently large, the Euler characteristic of the BRST complex is well-defined as an analytical function. We show that the degree of freedom is equal to minus the residue at infinity of
the BRST complex Euler characteristic logarithmic derivative.

We also present new facts about DoF for linear PDEs, including the result that conjugate linear systems (related by formal Hermitian conjugation of their differential operators) have  the same DoF under mild assumptions. 

The article is organized as follows. In the next section, we briefly outline the basics of linear gauge field theories, including 
the gauge symmetry transformations and gauge identities 
of PDE systems. 
The sequences of gauge identities and gauge symmetries 
of the system defines the BRST complex. 
We recall the generalities of the construction of the BRST complex for not necessarily Lagrangian systems \cite{KazLS:2005} in Subection~\ref{subs: BRST}.  For linear gauge field theory, we have a natural sequence of modules over the ring of polynomials of commutative variables $\partial_\mu$ being partial derivatives by space. The operators of equations, gauge symmetries, and gauge identities define the sequence of maps between these modules. This sequence is a complex that can be thought of as the Fourier transform of the configuration space for the BRST complex of the PDE system. The DoF count algorithm of the article \cite{KLS2013} is formulated in terms of this complex. 

In Section~\ref{sec:gauge_and _polynom}, we translate
the above consideration into the framework. 
of polynomial modules. First, we consider complexes of free 
polynomial modules with differentials defined by the operators of gauge identities, gauge symmetries, and their transposes. In this way, we clarify the connection with the BRST complex. Then (in Subsection~\ref{subs:N=e}) we prove the above formula for the Einsteinian DoF in terms of polynomial modules. 

In the next Subsection~\ref{subs: homogen}, we 
focus on the key case of homogeneous differential equations.
In this case,  the orders of the equations, gauge symmetries, and gauge identities are well-defined. We use it to prove that the Einsteinian DoF is equal to the one defined by the orders of gauge structures. Thus, we confirm the conjecture mentioned above for homogeneous systems. 
We conclude that for such systems,  DoF does not change after the conjugation of the system. 
Moreover, in Subsection~\ref{subs: BRST_DoF} we give a homological interpretation of the degree of freedom in terms of the BRST complex.  Besides the ghost number, we equip the complex with one more $\ZZ$-grading we call differential order. 
In Proposition~\ref{prop: DoF via BRST}  we compute Euler characteristics of the BRST complex with respect to this grading. We show that DoF equals to  minus the residue at infinity of the Euler characteristic logarithmic derivative.
Analogous formula in terms of the algebraic dual version of BRST complex involves the residue of the Euler characteristic logarithmic derivative
at zero, see the same Proposition~\ref{prop: DoF via BRST}.

The case of a general (non-homogeneous) system 
is considered in Section~\ref{sec: non-homogen}. 
We refine the notion of ``involutive system'' in the sense of~\cite{KLS2013}  by introducing the concepts of {\em weakly involutive} and {\em doubly weakly involutive} systems. Here ``doubly''  means that both the system and its conjugate are weakly involutive. For doubly weakly involutive systems, we prove the conjecture from~\cite{KLS2013}, and we simultaneously establish the equality of DoF for a pair of conjugate systems. We illustrate the formulae by  examples of the Maxwell equations (Examples~\ref{ex: maxwell} and~\ref{ex: maxwell via Euler}),
massive spin two field equations (Example~\ref{ex: massive spin 2}), and the Proca equations~(Example~\ref{ex: proca}).

In Appendix~\ref{app: algebra}, we collect some commutative algebra definitions and results necessary for the main part of the paper. The most important of them are the definition of the dimension and the multiplicity (see Proposition~\ref{prop:dim_n_mult}) and bounds for the dimension of $Ext$ modules in Proposition~\ref{prop: groth} (the last is mostly a part of the Grothendieck local duality theorem). Then, we connect the multiplicity of a graded module with its Euler characteristic in Proposition~\ref{prop:Q and multiplicity}. 

\section{Linear gauge systems}

\subsection{Configuration space and equations of motion}
Linear field theory begins with a certain space-time $X$, which is supposed to be a $d$-dimensional linear space. The space-time is sometimes referred to as the source 
space for the field theory. For certainty, we assume the source is the Minkowski space, $X={\mathbb{R}}^{1,d-1}$. The linear coordinates on $X$ are denoted $x^\mu, \mu=0,1,\dots ,d-1$, and
$x=(x^0,\dots , x^{d-1})\in X$.

Then, in linear field theory, there is a target space $M$ which is supposed to be an $m$-dimensional linear space. We denote the linear coordinates on $M$ by $\phi^i,\, i=1, \dots, m$, $\phi=(\phi^1, \dots\phi^m)\in M$. 

The fields $\phi(x)$ are the smooth functions that map $X$ to $M$, $\phi^i\in \mathcal{C}^\infty(X)$.
We denote by $\mathcal{M}$ the linear space of all smooth functions $\phi: X \mapsto M$. $\mathcal{M}$ is understood as the configuration space of field theory.

The true configurations of fields are defined as solutions of the field equations (also known as equations of motion, EoMs) being the system of $n$ linear PDEs with constant coefficients,
\begin{equation}
\label{EoMs}
  T_a\equiv \hat{T}_{ai}\phi^i=0 \, ,  \quad a=1, \dots, n,
\end{equation}
where the elements of matrix $\hat{T}_{ai}$ are the differential operators with constant coefficients,
\begin{equation}
\label{T-hat} 
\hat{T}_{ai}=\sum_{k=0}^{k_{ai}^{max}}
\sum_{\mu_1,\dots, \mu_k}T_{ai}^{\mu_1\dots\mu_k}\partial_{\mu_1} \dots \partial_{\mu_k}
 , \quad  T_{ai}^{\mu_1\dots\mu_k}=const, \qquad \partial_\mu=\frac{\partial}{\partial x^\mu} \, .
\end{equation}
The left hand side $T_a$ of the system (\ref{EoMs}) belongs to the linear space $\Bar{\mathcal{M}}$ of the $n$-tuples of smooth functions, $T:X\mapsto \bar{M}$, where $\bar{M}$ is an $n$-dimensional linear space. In other words, $T$ is a section of the rank $n$ trivial vector bundle over $X$.  We refer to $\bar{M}$ as the conjugate target space, while $\mathcal{\bar{M}}$ is termed a conjugate configuration space.

The linear space of solutions of the field equations (\ref{EoMs}), being the subspace of the configuration space $\mathcal{M}$, is usually referred to as the \emph{mass shell} of the field theory. We denote the mass shell as $\Sigma$:
\begin{equation}\label{Sigma}
  \Sigma=\{\phi\in\mathcal{M}\,|\,\hat{T}\phi=0\}\, .
\end{equation}

One more relevant notion for linear gauge systems is the formal Hermitian conjugation.
Given the  matrix differential operator $\hat{O} (\partial){}_{ai}$ whose entries are polynomials in partial derivatives $\partial_\mu$, the matrix of the Hermitian conjugate operator $\hat{O}{}^\dagger$ is defined by the rule
\begin{equation}\label{O-dagger}
    \hat{O}{}^\dagger{}_{ia}(\partial)=\hat{O}{}_{ai}(-\partial)\, .
\end{equation}
This means that the matrix of the Hermitian conjugate operator is defined by the transposition and reflection of the formal variables $\partial_\mu$.
The original matrix differential operator $\hat{O}$ maps the configuration space to the conjugate configuration space, while the Hermitian conjugate operator $\hat{O}{}^\dagger$ defines the inverse mapping of the corresponding dual spaces:
\begin{equation}
\label{OO-dagger}
\hat{O}{}\,:\, \mathcal{M}\,\mapsto\, \Bar{\mathcal{M}}\,,\quad\hat{O}{}^\dagger :\,  \Bar{\mathcal{M}}^*\,\mapsto\,\mathcal{M}^*\, .
\end{equation}
There is an important special case when $\Bar{\mathcal{M}}=\mathcal{M}^*$. In this case, both $\hat{O}{}$ and $\hat{O}{}^\dagger$ map $\mathcal{M}$ to $\mathcal{M}^*$.
Hence, any operator mapping these space decomposes into a Hermitian and anti-Hermitian part,
\begin{equation}\label{Herm-decomp}
   \hat{O}{}\,:\, \mathcal{M}\,\mapsto\, \mathcal{M}^*\, \quad\Rightarrow \quad\hat{O}=\hat{H} +\hat{A}, \quad \hat{H}=\hat{H}{}^\dagger, \,\, \hat{A}=-\hat{A}{}^\dagger.
\end{equation}
If $\Bar{\mathcal{M}}=\mathcal{M}^*$ and the operator $\hat{T}$ of the equations 
(\ref{EoMs}) is Hermitian, then the equations are Lagrangian, i.e. they follow from the least action principle.  Given the Hermitian matrix operator $\hat{T}=\hat{T}{}^\dagger$, it defines the bi-linear action functional
\begin{equation}\label{S}
    S[\phi(x)]=\frac{1}{2}\int dx \phi^i  \hat{T}_{ij}\phi^j\,.
\end{equation}
The Euler--Lagrange equations for this action reproduce the original field equations
(\ref{EoMs}). If the operator $\hat{T}$ in action (\ref{S}) was not Hermitian (while the number of fields is equal to the number of equations), the Euler--Lagrange equations would involve only the Hermitian part, while the anti-Hermitian part  drops out from the equations (\ref{EoMs}):
\begin{equation}\label{E-L-Eqs}
    \frac{\delta S}{\delta \phi^i}=\frac{1}{2}\left(\hat{T}_{ij}+ \hat{T}{}^\dagger_{ij}\right)\phi^j \, .
\end{equation}


Two systems of linear field equations (\ref{EoMs}) say $T_a=0$ and $T'_{a'}=0$, are said to be equivalent if they are differential consequences of each other. This means, the differential operators $\hat{D}^a_{a'}$ and ${}'\!\hat{D}_a^{a'}$
exist such that
\begin{equation}\label{TTprime}
  \hat{D}^{a}_{a'} T_{a}=T_{a'}, \, {}'\!\hat{D}_{a}^{a'}T_{a'}=T_a  , \quad a=1,\dots , n,\, a'=1,\dots , n'\,.
\end{equation}
In general, $n\neq n'$, i.e. the number of equations can be different for equivalent systems. For example, if any set of differential consequences complements the system of equations, it is an equivalent system. Obviously, the solutions coincide for equivalent systems of field equations.

Let us discuss ordes of the differnetial equations $T_a$. 
We assign an integer differential order $\theta_i$ to each field variable
$\phi^i$ ($i=1, \dots, m$) in~(\ref{EoMs}). It can be thought of as the inverse dimension of length. 
In the simplest situation (which we refer to as {\em standard grading}), all 
$\theta_i$ are equal to zero. 
We assume that each partial differentiation $\partial_\mu$ increases the dimension of variables by 1, it increases the differential order by 1.
Then each differential operator assigms its differential order as $\ord \phi^i = \theta_i$, $\ord (\partial_\mu \phi^i) = \theta_i +1$, etc. 
By the order of the specific differential equation $T_a=0$  of the system (\ref{EoMs}) we understand the maximal order of the partial derivatives involved in the equation:
\begin{equation}\label{OrderT}
  \ord(T_a)= \max_i\{k_a^i + \theta_i \}\, .
\end{equation}
In the case of standard grading, we have $ \ord(T_a)=\max_i\{k_a^i\}$. We denote $k_a = \ord T_a$. 

Let us explain now the notion of gauge symmetry transformation for the system of linear field equations (\ref{EoMs}).
Consider a set of differential operators $\hat{R}^i$ that maps the linear space of smooth functions to the configuration space of fields:
\begin{equation}\label{R}
  \phi^i_\epsilon=\hat{R}{}^i\epsilon(x)\in\mathcal{M}, \quad\hat{R}{}^i=\sum_{k=0}^{r_i}R^{i\,\mu_1\dots\mu_k}\partial_{\mu_1}\dots\partial_{\mu_k},\quad R^{i\,\mu_1\dots\mu_k}\in\mathbb{C}, \quad\epsilon\in\mathcal{C}^\infty(X).
\end{equation}
The arbitrary function $\epsilon$ being involved in this map is called the gauge transformation parameter.

For example, let the fields be the components of the one-form $A_\mu$. The operator $\hat{R}{}_\mu$ can be the partial derivative $\partial_\mu$ that maps the smooth functions to one-forms, $A_\mu^\epsilon=\partial_\mu\epsilon$.
For another example, let us consider the fields to be the components of a symmetric second rank tensor $h_{\mu\nu}=h_{\nu\mu}$ on the Minkowski space.
Consider the operator $\hat{R}_{\mu\nu}=\partial_\mu\partial_\nu$ being the second partial derivative. It maps any smooth function $\epsilon$ to the space of symmetric tensors $h_{\mu\nu}^\epsilon=\partial_\mu\partial_\nu\epsilon$.

The linear map (\ref{R}) from the space of smooth functions to the configuration space of fields is said to be a gauge symmetry transformation if $\phi^i_\epsilon$ is a solution of the field equations (\ref{EoMs}) for arbitrary $\epsilon\in\mathcal{C}^\infty$:
\begin{equation}\label{GS}
 \hat{T}_{ai}\phi^i_\epsilon=0\,, \quad\forall\epsilon\in\mathcal{C}^\infty(X)\, .
\end{equation}
The solution of the field equations is considered trivial if it is a gauge transformation for any specific gauge parameter. 

Any two solutions are considered equivalent if their difference is trivial:
\begin{equation}\label{equiv}
  \phi^i\sim\phi'{}^i\quad\Leftrightarrow\quad \exists\epsilon\in\mathcal{C}^\infty(X):\,\,  \phi^i-\phi'{}^i=\hat{R}^i\epsilon
\end{equation}
The operator $\hat{R}^i$ from (\ref{R}) generating a trivial solution (\ref{GS}) is not unique. 
For example, the action of any differential operator with constant coefficients on a trivial solution results in another trivial solution. In general, there can be different trivial solutions which do not necessarily reduce to the derivatives of the unique trivial solution.

Let us assume that field equations (\ref{EoMs}) admit the set of $m_1$ different trivial solutions
\begin{equation}\label{triv-set}
\phi^i_1=\hat{R}{}^i_1\epsilon^1,\, \dots, \, \phi^i_{m_1}=\hat{R}{}^i_{m_1}\epsilon^{m_1}
\end{equation} 
of the field equations, i.e. for every $\alpha=1,\dots , m_1$ the field $\phi^i_\alpha$ is a gauge symmetry transformation (\ref{GS}).
Any linear combination of trivial solutions and their derivatives is considered a trivial solution,
\begin{equation}\label{lin-comb}\phi^i_{(\epsilon_1\dots\epsilon_{m_1})}=\hat{R}^\alpha\phi^i_\alpha,
\end{equation}
where $\hat{R}^\alpha, \alpha=1,\dots, m_1$ are differential operators with constant coefficients.

The trivial solutions $\phi_\alpha^i, \alpha=1,\dots , m_1 $ defined in  (\ref{triv-set})
 are said to form a generating set of gauge symmetry transformations if any trivial solution (\ref{R}), (\ref{GS}) reduces to a linear combination of solutions  (\ref{triv-set}) and their derivatives:
\begin{equation}\label{G-set}
  \hat{T}_{ai}\hat{R}{}^i\epsilon=0,\,\,\forall\epsilon\in\mathcal{C}^\infty(X)\quad\Leftrightarrow\quad\exists\hat{R}{}^\alpha\,: \hat{R}{}^i=\hat{R}{}^\alpha\hat{R}{}^i_\alpha\, ,
\end{equation}
where $\hat{R}{}^\alpha$ are differential operators with constant coefficients.

The set of gauge transformation parameters $\epsilon^\alpha, \, \alpha=1,\dots, m_1$ 
maps the space-time $X$ to the linear space $M_1, \, \dim M_1=m_1$. 
The gauge parameters $\epsilon$ belong to the space of smooth functions $\mathcal{M}_1$, $\epsilon\,:X\,\mapsto M_1$.  We refer to $M_1$ as the target space of the gauge parameters, and $\mathcal{M}_1$ is said to be the configuration space of the gauge parameters.

The general trivial solution is understood as the linear combination of the elements of the generating set of trivial solutions involving independent arbitrary functions $\epsilon^\alpha$:
\begin{equation}\label{G-Triv}
  \phi^i_\epsilon=\hat{R}{}^i_\alpha\epsilon^\alpha.
\end{equation}
Any trivial solution of the field equations is also known as ``a pure gauge", or a gauge transformation of the fields.  Any pure gauge solution can be derived from the general trivial solution (\ref{G-Triv}) by selecting specific gauge parameters $\epsilon^\alpha=\hat{R}{}^\alpha\epsilon$, cf. (\ref{G-set}). 

Now, let us explain the physical reasons behind the idea to consider the solutions of the field equations (\ref{EoMs}) as trivial if they reduce to a linear combination of the derivatives of the arbitrary functions $\epsilon\in\mathcal{C}^{\infty}(M)$. The reason is that the observable classical dynamics are assumed to be uniquely determined by specifying the initial data and/or boundary conditions. This determinism is understood as the classical causality of dynamics. The initial or boundary conditions restrict the values of fields $\phi^i (x), x\in X$ on certain submanifold $
Y
\subset X, \, \dim 
Y
<\dim X$. The values of the initial/boundary data should unambiguously determine the mass shell. The fields on $
Y
\subset X$ are the functions of less than $d$ coordinates on $X$. Imposing any boundary or initial conditions, one cannot restrict the gauge parameters $\epsilon\in\mathcal{C}^\infty (M)$ involved in the pure gauge solution (\ref{G-Triv}) as it is a solution with arbitrary $\epsilon$. Since the pure gauge solutions involve arbitrary functions of all  $d$ coordinates, these solutions cannot contribute to any 
physical observable
being a function of the fields and their derivatives.  To put it  simply, arbitrarily evolving quantities cannot be observable. Since the different solutions of the field equations define the same observable quantities, they describe the same physical reality. Because of this  any two solutions to the field equations (\ref{EoMs}) are considered equivalent if their difference is a gauge transformation.

\subsection{Einsteinian strengths of equations and degree of freedom}

\label{subs:deg_freedom_definition}

Einstein~\cite{einstein} proposed a measure for the degree of freedom admitted by the system~(\ref{EoMs}). 
Intuitively, the measure counts the number of independent 
functions of $d-1$ independent variables needed to define a general solution of the system~(\ref{EoMs}) modulo arbitrary functions of $d$ variables. The latter are related to nonphysical ambiguity in the solutions that originate from the gauge symmetry, see \cite{einstein}. In the case of variational systems, the equations can be brought to Dirac's Hamiltonian constrained form \cite{Dirac} that allows for another way to systematically count the DoF. For not necessarily Lagrangian systems the Dirac-Bergmann method is extended in the article \cite{Lyakhovich:2008hu}.

Let us explain Einsteinian DoF count.
The field
$\phi:X\mapsto M $ is expanded as a Taylor series
\begin{equation}\label{eq:phi_Taylor_expansion}
\phi(x) =
\sum_{\alpha = (\alpha_0, \dots, \alpha_{d-1}), \alpha_i \ge 0} 
\phi_\alpha x^\alpha .
\end{equation}
Given $N\ge 0$,  consider the projection $p_N$
which sends the above function $\phi(x)$ into the polynomial 
$$ p_N(\phi(x)) = \sum_{\alpha\ge 0,\, \alpha_0+\dots +\alpha_{d-1} \le N}  \phi_\alpha x^\alpha =
\phi(x)+o(|x|^N).$$

Let $\Sigma$ denote the set of all solutions of the system~(\ref{EoMs}). Then its image $p_N(\Sigma)$ under $p_N$ is a vector space of polynomials of degree at most $N$. The integral function $h_{\Sigma} (N) = \dim p_N(\Sigma) - \dim p_{N-1}(\Sigma)$ 
(called the Hilbert function of $\Sigma$) 
measures the growth of the solution space. Essentially, its value is equal to the number of independent $N$-th order coefficients in the Taylor series expansion of the general solution. 
For example, if the system is trivial (that is, all coefficients of the operators $T_a$ are zero%
), we have $h_{\Sigma} (N) = m\binom{N+d}{N}- m\binom{N+d-1}{N-1} = m\binom{N+d-1}{N}$. Similarly, if ${\Sigma'}\subset \Sigma$ denotes the set of all trivial solutions, the corresponding  Hilbert function $h_{\Sigma'} (N) = \dim p_N(\Sigma') - \dim p_{N-1}(\Sigma')$ measures the growth of the space of trivial solutions. Thus, the growth of the physically meaningful (that is, non-trivial) solutions is measured by the difference
$h_{\Sigma} (N) - h_{\Sigma'} (N) $.
Finally, the {\em physical degree of freedom} for the system~(\ref{EoMs})  is defined as
\begin{equation}
\label{eq:deg_free_definition}
{\mathcal N} = \frac{1}{d-1} \lim_{N\to \infty} N \frac{h_{\Sigma} (N) - h_{\Sigma'} (N)}{\binom{N+d-1}{N}}. 
    \end{equation}
According to Einstein, the larger the physical degree of freedom, the weaker is the system.
Note that Einstein has used the term {\em coefficient of freedom} for the value 
$$
z_1 = (d-1 ){\mathcal N} =  \lim_{N\to \infty} N \frac{h_{\Sigma} (N) - h_{\Sigma'} (N)}{\binom{N+d-1}{N}}. $$

For example, for the trivial system $0=0$ we have zero coefficient of freedom, since in this case $h_{\Sigma} (N) = h_{\Sigma'} (N) $. Einstein 
shows by counting the equation strengths \cite{einstein} 
  that for both Maxwell's and Einstein's equations the degree of freedom 
is equal to 4. 
Note that the Hilbert function difference 
$h_{\Sigma} (N) - h_{\Sigma'} (N)$ for these two theories are not equal, so that
the strength of these two theories is the same only at the local level. 

\subsection{Further gauge transformations and gauge identities}

Let us consider the equations
\begin{equation}\label{Triv0}
\hat{R}{}^i_\alpha\epsilon^\alpha=0,
\end{equation}
where the gauge parameters $\epsilon^\alpha$ are treated as unknown fields. These equations state that with such gauge parameters, the gauge transformation vanishes of the original fields $\phi^i$. 
If these equations do not admit gauge symmetry of their own, the generating set (\ref{G-set}) of gauge transformations is said to be irreducible.
 The irreducibility condition reads
 \begin{equation}\label{Gauge-irred}
   \hat{R}{}^i_\alpha\hat{R}{}^\alpha\epsilon_{(1)}=0, \, \forall\epsilon_{(1)}\in\mathcal{C}^\infty(X)\quad\Leftrightarrow\quad\hat{R}{}^\alpha=0\, .
 \end{equation}

 If the equations (\ref{Triv0}) admit gauge symmetry of their own, the gauge symmetry (\ref{G-Triv}) of original field equations (\ref{EoMs}) is said reducible.

Let us consider the general trivial solution of the equations (\ref{Triv0}),
\begin{equation}\label{G-for-Gauge-transform}
\epsilon^\alpha=  \hat{R}_{(1)}{}^\alpha_{\alpha_1}\epsilon_{(1)}^{\alpha_1} \,, \quad \alpha_1=1,\dots, m_1 .
\end{equation}
The arbitrary functions  $\epsilon_{(1)}^{\alpha_1}$ involved in the solution (\ref{G-for-Gauge-transform}) are said to be the first stage gauge parameters for the original gauge symmetry parameters. The solution above of  equations (\ref{Triv0}) is also known as a gauge for gauge symmetry transformation. These gauge for gauge transformations can be reducible or irreducible in the same sense as the gauge transformations of the original fields. If the gauge for gauge transformations (\ref{G-for-Gauge-transform}) are reducible, the equations 
\begin{equation}\label{G-for-G2}
\hat{R}_{(1)}{}^\alpha_{\alpha_1}\epsilon_{(1)}^{\alpha_1} = 0
\end{equation}
admit general trivial solution
\begin{equation}\label{G2-for-Gauge-transform}
\epsilon_{(1)}^{\alpha_1}=  \hat{R}_{(2)}{}^{\alpha_1}_{\alpha_2} \epsilon_{(2)}^{\alpha_2} \,, \quad \alpha_1=1,\dots, m_2 .
\end{equation}
These gauge transformations of the second stage can be either irreducible or reducible once again. The sequence of gauge transformations of gauge transformations ends at some final stage $k$, when the symmetry turns out to be irreducible. The length of the sequence can depend on the choice of the generating set of the gauge transformations at each stage. The generating set can be chosen at each stage in such a way that the length of the sequence is finite with $k\leq d$. We explain this fact later.

Let us define the order of the gauge for gauge transformations (\ref{G-for-Gauge-transform}), (\ref{G2-for-Gauge-transform}). For any gauge parameter $\epsilon^\alpha$ of the generating set (\ref{G-Triv}) of gauge transformations, the differential order $r_\alpha$ is defined similarly to (\ref{OrderT}),
\begin{equation}\label{r-alpha}
r_\alpha=\max_{i}{(r^i_\alpha + \theta_i)} \, ,
\end{equation}
where $r^i_\alpha$ is the maximal order of the derivatives in the operator $\hat{R}{}^i_\alpha$, c.f. (\ref{R}).
For the standard grading $\theta_i =0 $, one gets therefore $r_\alpha=\max_{i} r^i_\alpha$.

The order $r_{\alpha_1}^{(1)}$ of the gauge for gauge transformation with parameter $\epsilon^{\alpha_1}_{(1)}$ is defined as follows:
\begin{equation}\label{r-alpha1}
    r_{\alpha_1}^{(1)}=\max_{\alpha}{(r_{\alpha_1}^\alpha+r_\alpha)} \, ,
\end{equation}
where $r_{\alpha_1}^\alpha$ is a maximal order of the derivatives in the operator $\hat{R}{}_{\alpha_1}^\alpha$, c.f. (\ref{R}).

Let us consider the general gauge for gauge transformation of the stage $k$, with gauge parameters $\epsilon^{\alpha_k}$ generated by the set of differential operators
$\hat{R}{}_{(k)\alpha_k}^{\alpha_{k-1}}, \,  \alpha_k=1,\dots,m_1$ of  order 
$r_{\alpha_k}^{(k)\alpha_{k-1}}$, 
\begin{equation}\label{k-stage}
  \hat{R}{}^{(k-2)\alpha_{k-2}}_{\alpha_{k-1}} \hat{R}{}_{{(k)\alpha}_k}^{\alpha_{k-1}}\epsilon^{\alpha_k}=0, \quad \forall\epsilon^{\alpha_k}\in\mathcal{C}^\infty(M)\, .
\end{equation}
The order of gauge for gauge transformation above is inductively defined, starting from $k=1$ (\ref{r-alpha1}):
\begin{equation}\label{r-alpha-k}
    r_{\alpha_k}^{(k)} =
    \max_{\alpha_{k-1}}{(r_{(k)\alpha_k}^{\alpha_{k-1}}
    +r_{\alpha_{k-1}}^{(k-1)})  } \, ,
\end{equation}
For completeness, we denote $r^{(0)}_{\alpha} = r^i_\alpha$ and $\alpha_0=\alpha$.

Now, let us discuss gauge identities between the field equations (\ref{EoMs}). By gauge identities we understand the linear differential relations between the left hand sides of the field equations which holds true for any fields:
\begin{equation}\label{L-A}
    \hat{L}{}^a_A \hat{T}_{ai}\phi^i\equiv 0\,, \quad\forall\phi\in\mathcal{E}\,\quad A=1,\dots ,l.
\end{equation}
The above set of $l$ gauge identities is not necessarily unique. The differential operators $\hat{L}{}^a _A, \, A=1,\dots, l$ are said to form the generating set for gauge identities if any other identity is a differential consequence of the identities (\ref{L}) :
\begin{equation}\label{L}
    \hat{L}{}^a \hat{C}_{ai}\phi^i\equiv 0\,, \quad\forall\phi\in\mathcal{E}\quad\Leftrightarrow\quad\exists\,\hat{L}{}^A:\hat{L}{}^a\phi_a=\hat{L}{}^A\hat{L}{}^a_A\phi_a, \, \forall\phi_a\in\mathcal{C}^\infty(M) \, .
\end{equation}
The set of gauge identity generators $\hat{L}{}^a_A$ (\ref{L-A}) is said irreducible if the equations
\begin{equation}
\label{L-zero}
\hat{L}{}^a_A\phi_a =0
\end{equation}
do not admit gauge identities of their own. In the opposite case, the gauge identities (\ref{L-A}) are said reducible. The reducibility of identities means, there exist differential operators $\hat{L}_{(1)}{}^A_{A_1}$ such that generate identities for equations (\ref{L-zero}):
\begin{equation}
\label{L-1}
    \hat{L}_{(1)}{}^A_{A_1}\hat{L}{}^a_A\phi_a =0\, ,\quad \forall\phi^a\in\mathcal{C}^\infty (M).
\end{equation}
These identities can be reducible again in the sense that there can exist identities between the identities. One can see that  equations (\ref{EoMs}) can have the sequence of identities of identities much like they can admit gauge for gauge symmetries.
The identities for identities of the stage $k$ read
\begin{equation}\label{L-k}
    \hat{L}_{(k)}{}^{A_{k-1}}_{A_k}\hat{L}_{(k-1)}{}^{A_{k-2}}_{A_{k-1}}\phi_{A_{k-2}} =0\, ,\quad \forall\phi_{A_{k-2}}\in\mathcal{C}^\infty (M) .
\end{equation}
The generating set of gauge identities can be chosen at each stage in such a way that the length of the sequence is finite with $k\leq d$. We explain this fact later along with a similar fact concerning the sequence of gauge symmetries for symmetries.

The order $l^{(0)}_A$ of gauge identities (\ref{L}) is defined by the orders $l^a_A$ of the differential operators $\hat{L}{}_A^a$ and the orders (\ref{OrderT}) of the equations of motion (\ref{EoMs}):
\begin{equation}\label{L-order0}
    l^{(0)}_A=\max_a({l^a_A+\ord({T}_a)}).
\end{equation}
The order of the $k$-th stage gauge identity (\ref{L-k}) is iteratively defined by the rule:
\begin{equation}\label{L-orderk}
    l^{(k)}_{A_k}=\max_{A_{k-1}}\left(l^{A_{k-1}}_{A_k}+l^{(k-1)}_{A_{k-1}}\right).
\end{equation}

The generators of gauge identities $\hat{L}{}^a_A$, being the matrix differential operators, map the  conjugate configuration space  $\Bar{\mathcal{M}}$ to the space $\Bar{\mathcal{M}}_1$ of $\bar{m}_1$-tuples of smooth functions $\bar{\phi}{}^A$.
The functions $\bar{\phi}{}^A$ map the space-time $X$ to linear space $M_1, \dim {\Bar{M}_1}=\bar{m}_1$. We refer to ${\Bar{M}_1}$ as the conjugate gauge target space.
This terminology and notation naturally extends to identities of identities. The generators  of $\bar{k}$-th stage gauge identity (\ref{L-k}) map the $\bar{k}-1$ stage conjugate gauge configuration space $\Bar{\mathcal{M}}{}_{\bar{k}-1}$ to the conjugate configuration space $\Bar{\mathcal{M}}{}_{\bar{k}}$  of the next stage.
Correspondingly, the elements of  $\Bar{\mathcal{M}}_{\bar{k}}$ 
are the $\bar{m}_{\bar{k}}$-tuples of smooth functions that map the space-time $X$ 
to the $\bar{k}$-th stage conjugate target space $\bar{M}_{\bar{k}}, \,\dim\bar{M}_{\bar{k}}=\bar{m}_{\bar{k}}$. We obtain a sequence of maps
\begin{equation}\label{Gauge-for-Gauge-Sequence}
   0\leftarrow \Bar{\mathcal{M}}_{\Bar{k}_{max}} \leftarrow\cdots\stackrel{\hat{L}_{(\bar{k})}}{\longleftarrow}\cdots\leftarrow
\Bar{\mathcal{M}}_{1}\stackrel{\hat{L}}{\longleftarrow}\Bar{\mathcal{M}}\stackrel{\hat{T}}{\longleftarrow}\mathcal{M}\stackrel{\hat{R}}
   {\longleftarrow}\mathcal{M}_{1}{\leftarrow} \cdots \stackrel{\hat{R}_{(k)}}{\longleftarrow}
    \cdots{\leftarrow}\mathcal{M}_{\Bar k_{max}}\leftarrow 0\,.
\end{equation}
The image of every map in the sequence belongs to the kernel of the next map. The sequence is not exact, in general, in the segment 
toward the left of the map $\hat T$
as the the pure gauge solutions do not necessarily exhaust the mass shell, and the LHS of the field equations do not necessarily span the kernel of the gauge identity generators.

The formal Hermitian conjugation of all the operators in the sequence (\ref{Gauge-for-Gauge-Sequence}) results in the conjugate sequence
\begin{equation}\label{Gauge-for-Gauge-Sequence-Conjugate}
   0\rightarrow \Bar{\mathcal{M}}^*_{\Bar {k}_{max}} 
   \rightarrow\cdots\stackrel{\hat{L}{}^{\dagger}_{(\bar{k})}}{\longrightarrow}\cdots\rightarrow\cdots \Bar{\mathcal{M}}^*_1
   \stackrel{\hat{L}{}^{\dagger}}{\longrightarrow}
 \Bar{\mathcal{M}}^*
\stackrel{\hat{T}{}^{\dagger}}{\longrightarrow}
\mathcal{M}^* \stackrel{\hat{R}{}^{\dagger}}
   {\longrightarrow}\mathcal{M}^*_{1}{\rightarrow} \cdots \stackrel{\hat{R}{}^{\dagger}_{(k)}}{\longrightarrow}
    \cdots{\rightarrow}\mathcal{M}^*_{k_{max}}\rightarrow 0\,.
\end{equation}
For Lagrangian systems (\ref{E-L-Eqs}), $\Bar{\mathcal{M}}=\mathcal{M}^*$, i.e. the left hand sides of the equations (\ref{EoMs}) are the elements of the dual space to the configuration space, and the matrix operator of equations $\hat{T}$ is Hermitian. This leads to the fact that the generators of gauge symmetries $\hat{R}$ and gauge identities $\hat{L}$ are connected by Hermitian conjugation, 
\begin{equation}\label{R-L}
    \phi^i\hat{T}_{ij}\hat{R}{}^j\epsilon =0, \quad \epsilon\hat{L}^i\hat{T}_{ij}\phi^j=0, \quad\forall\phi^i\in\mathcal{M}, \,\epsilon\in\mathcal{C}^\infty,\quad \hat{T}= \hat{T}{}^\dagger \quad \Rightarrow \quad\hat{R}^\dagger=\hat{L}\, .
\end{equation}
This pairing between gauge symmetries and gauge identities in Lagrangian systems can be viewed as a form of the second Noether theorem.
Given the fact that $\hat{R}^\dagger=\hat{L}$ for Lagrangian systems, the sequence of the gauge symmetries and identities (\ref{Gauge-for-Gauge-Sequence}) coincides with the conjugate sequence (\ref{Gauge-for-Gauge-Sequence-Conjugate}) in this case:
\begin{equation}\label{Gauge-for-Gauge-Sequence-Lagrangian}
   0\leftarrow \mathcal{M}^*_{{k}_{max}} \leftarrow\cdots\stackrel{\hat{R}^\dagger_{(\bar{k})}}{\longleftarrow}\cdots\leftarrow
{\mathcal{M}^*}_{1}\stackrel{\hat{R}^\dagger}{\longleftarrow}{\mathcal{M}^*}\stackrel{\hat{T}}{\longleftarrow}\mathcal{M}\stackrel{\hat{R}}
   {\longleftarrow}\mathcal{M}_{1}{\leftarrow} \cdots \stackrel{\hat{R}_{(k)}}{\longleftarrow}
    \cdots{\leftarrow}\mathcal{M}_{k_{max}}\leftarrow 0\,.
\end{equation}
This sequence can be split in two segments sharing the central element $\hat{T}:\mathcal{M}\mapsto\mathcal{M}^*$ which corresponds to the map of the fields into the left hand sides of the Lagrangian equations (\ref{E-L-Eqs}). The right segment includes the gauge symmetries of the equations,  and the sequence of symmetries for symmetries. This sub-sequence ends up with the map $\hat{T}:\mathcal{M}\mapsto\mathcal{M}^*$. The left segment of the sequence begins with the same map, and it continues to the left with gauge identities between the left hand sides of the equations, the identities of identities, etc.    
These two sub-sequences mirror each other in Lagrangian case because they are connected by Hermitian conjugation. That is why in the theory of Lagrangian gauge systems \cite{HT94}, only one sub-sequences matters as another sub-sequence is just the conjugation of the first one. For non-Lagrangian equations, there is no pairing between gauge symmetries and gauge identities. And these two sub-sequences of sequence (\ref{Gauge-for-Gauge-Sequence})  are not conjugate to each other.

\subsection{Coefficient of freedom in terms of the orders of gauge transformations and identities  }

The above orders of gauge transformations and gauge symmetries are used in the following formula for the physical degree of freedom. This formula has been proposed in~\cite{KLS2013} as a conjecture.  Suppose that the system~(\ref{EoMs}) is {\em homogeneous}, that is, in each $a$-th equation $T_a$(for $ a=1, \dots, n$), each term has the same order $k_a$, 
$\hat{T}_{ai}=
\sum_{\mu_1,\dots, \mu_{k_a -\theta_i}} \partial_{\mu_1} \dots \partial_{\mu_{k_a-\theta_i}}
T_{ai}^{\mu_1\dots\mu_{k_a-\theta_i}}.  
$
In the case of standard grading, this means that the linear operator $\hat T$ is defined by a homogeneous differential polynomial matrix. 
Then the linear operators $\hat{R}_{(k)}$ and $\hat{L}_{(k)}$ can be chosen homogeneous in the following sense: each term  in the differential polynomial $\hat{R}{}_{(k)\alpha_k}^{\alpha_{k-1}}$ (respectively, $\hat{L}_{(k)}{}^{A_{k-1}}_{A_k}$) has order $r{}_{(k)\alpha_k}^{\alpha_{k-1}}$ (respectively,
${l}_{(k)}{}^{A_{k-1}}_{A_k}$).

\begin{theorem}
\label{th:main_homogeneous}
Suppose that the system~(\ref{EoMs}) is homogeneous.
In the notation above, its physical degree of freedom is equal to the sum
	$${\mathcal N} = \sum_{a} k_a -
    \sum_i \theta_i -
    \sum_{k= 0}^{\Bar {k}_{max}} (-1)^k \sum_{A_k} l^{(k)}_{A_k} - 
 \sum_{k= 0}^{{k}_{max}} (-1)^k \sum_{\alpha_k}  r_{\alpha_k}^{(k)}
	. 
	$$

If the grading is standard, the formula have the form~\cite{KLS2013}:
	$${\mathcal N} = \sum_{a} k_a-\sum_{k= 0}^{\Bar {k}_{max}} (-1)^k \sum_{A_k} l^{(k)}_{A_k} - 
 \sum_{k= 0}^{{k}_{max}} (-1)^k \sum_{\alpha_k}  r_{\alpha_k}^{(k)}
	. 
	$$  
\end{theorem}

This theorem will be proved later, see Corollaries~\ref{cor:main_homogeneous_general} and~\ref{cor:main_homogeneous}.
Moreover, this formula holds for non-homogeneous systems under mild assumptions  ({\em doubly weakly involutive} systems), see Proposition~\ref{prop:Hilbert_pols_for_complexes}.

\subsection{BRST complex}

\label{subs: BRST}
The sequence (\ref{Gauge-for-Gauge-Sequence}) corresponds to the BRST complex related
to the system of equations (\ref{EoMs}). Construction of this complex begins with replacing the sequence of the target spaces $\Bar{M}_{\Bar{k}}, \, M_k, \, k=1,\ldots , K, \, \Bar{k}=1,\ldots , \Bar{K} $ by $\ZZ$ and ${\ZZ}_2$ graded linear spaces $\Bar{V}_{\Bar{k}} ,\, V_k$ of the same dimensions.  ${\ZZ}$-grading is known as a ghost number which is assigned to every linear coordinate. 

Another ${\ZZ}$-grading  is induced by the differential order. Recall that we have assigned differenatial order $\theta_i$ to each field variable $\phi^i$. This assignment (reflecting the dimensions of the field variables) uniquely defines the differential order values for further gauge structure by~(\ref{r-alpha}, \ref{r-alpha1}, 
\ref{r-alpha-k}, \ref{L-order0}, \ref{L-orderk}).

The ghost number distribution and the differential orders of the variables are arranged in the table below.
\begin{table}[ht]
\caption{}
\begin{center}
\begin{tabular}{| l | c | c | c | c | c | c | c | c |}
\hline
Graded target space &$ \overline{V}_{\overline{k}_{max}}$  & $ \cdots$ & $\phantom{'} \overline{V}_1$ \phantom{'}&$ \overline{V}_0$ & $V_0$  & $V_1$ & $\cdots$ & $ V_{K}$ \\ 
\hline
Linear coordinates 
& $\overline{\phi}{}_{A_{\overline{k}_{max}-1}}$ &    
$ \cdots$ &  
$\overline{\phi}{}_{A_0}$  & 
$\overline{\phi}{}_a$  
& $\phi^i$
&$\phantom{{}^{(1)}}$$\phi^{\alpha_1}$\phantom{@} & 
$\cdots$ & $\phi^{\alpha_{k}}$\phantom{@}\\ 
\hline
Ghost number & $-\bar{k}_{max}-1$ & $\cdots$& $-2$ & $-\,1$ & $0$ & $1$& $\cdots$& $k_{max}$\\
\hline
Differential order & $ l^{(\bar{k}_{max})-1}_{a_{\bar{k}_{max}-1}}$ & $\cdots$& $ l^{(0)}_{A}$ & $k_a$ & $\theta_i$ & $r_{\alpha}$& $\cdots$& $r_{\alpha_{k_{max-1}}}^{(k_{max}-1)}$\\\hline
\end{tabular}
\end{center}
\end{table}

We adopt the usual notation of general theory of gauge field theories \cite{HT94} where the ghost number is denoted by the symbol $\gh$, while the Grassmann parity is denoted by $\varepsilon$. Grassmann parity is defined by the ghost number modulo 2. For example,
\begin{equation}\label{gh}
    \gh(\phi^{\alpha_k})=k\, ,\quad \varepsilon(\phi^{\alpha_k})=\frac{1}{2}\left(1+ (-1)^{k+1}\right) \, , \qquad \gh(\phi^i)=\varepsilon(\phi^i)=0\, , \quad \gh(\Bar{\phi}_a)=-1\,, \quad \varepsilon(\Bar{\phi}_a)=1\, .
\end{equation}
The original target space $M$ is identified with $V_0$. In contrast, the graded target space $\Bar{V}_0$ is defined as the original conjugate target space $\Bar{M}$ with shifted Grassman parity and ghost number $-1$. The ghost number of a product is defined as the sum of the factors' ghost numbers, so any monomial of the graded coordinates has a certain ghost number. 

Given the graded target spaces arranged in the table, we can consider the graded fields 
that map the space-time $X$ to the corresponding vector space. In this way, we arrive at the spaces of graded fields $\mathcal{V}_k\ni\phi^{\alpha_k}:X\mapsto V_k,\,\Bar{\mathcal{V}}_{\bar{k}}\ni\bar{\phi}_{a_{\bar{k}}}:X\mapsto \bar{V}_{\bar{k}}$.
The fields with non-negative ghost numbers are termed just as fields, while the fields with negative ghost numbers are termed anti-fields. The fields with strictly positive ghost numbers are termed as ghosts. For Lagrangian system the fields and anti-fields belong to the dual spaces \cite{HT94}, while for non-Lagrangian equations, there are no pairing between fields and anti-fields \cite{KazLS:2005}.

The field-anti-field configuration space of the BRST  theory $\mathcal{W}$ is a 
direct sum of all the spaces of graded fields: 
$\mathcal{W}=\mathcal{V}_K\oplus \cdots \oplus\mathcal{V}_0 \oplus \Bar{\mathcal{V}}_0 \oplus \cdots \oplus\Bar{\mathcal{V}}_{\Bar{K}}$. 

The space-time coordinates have zero ghost number, so the ghost number of the derivative with respect to any  $x^\mu$ does not change the ghost number of any field or anti-field.

 The BRST differential $Q$ is the Grassmann-odd vector field  with ghost number 1, which acts on the field-antifield configuration space $\mathcal{W}$,
 \begin{equation}\label{Q-grad}
     Q:\, \mathcal{W}\mapsto \mathcal{W}, \quad \gh (Q f) =\gh(f)+1, \quad \varepsilon (Q f) = 
     \varepsilon ( f) + 1 \, \mod 2\, .
 \end{equation}
 The action of the BRST differential on the field-anti-field variables reads
\begin{equation}\label{Q}
    Q\phi^i=\hat{R}^i_{\alpha_1}\phi^{\alpha_1}, \quad   Q\phi^{\alpha_k}=\hat{R}^{\alpha_k}_{\alpha_{k+1}}\phi^{\alpha_{k+1}},\qquad Q\bar{\phi}_a=\hat{T}_{ai}\phi^i\, , \quad Q{\bar{\phi}}_{A_{\bar{k}+1}}= \hat{L}{}^{A_{\bar{k}}}_{A_{\bar{k}+1}}{\bar{\phi}}_{A_{\bar{k}}}.
     \end{equation}

Given the definitions of generators of gauge symmetries (\ref{r-alpha}), (\ref{r-alpha-k}) and gauge identities (\ref{L-A}), (\ref{L-k}), one can see that $Q$ squares to zero.
This equips the  field-anti-field configuration space 
$\mathcal{W}$ with the structure of complex. 
   The BRST complex is the free graded commutative algebra generated by the differential graded module $\mathcal{W}$.   
 The action of $Q$ is extended from the fields and anti-fields to polynomials of their derivatives by the Leibnitz rule.

The BRST cohomology groups are naturally graded by the ghost number. These groups admit natural physical interpretation. In particular, zero ghost number cohomology group elements are interpreted as equivalence classes of gauge invariants of the original system. Interpretation of some other important BRST cohomology groups with the other ghost numbers can be found in the review \cite{BarnichBrandtHenneaux2000} for Lagrangian equations. For study of the BRST cohomology groups
for general PDE's we refer to the article \cite{KapLS2011}.

In general, the BRST differential 
is not homogeneous with respect to the differential order. 
Nevertheless, in the case of homogeneous system the BRST differential $Q$ is indeed homogeneous with zero differential order. We use this grading to define Euler characteristic of BRST complex and to calculate the degree of freedom using this Euler characteristic, see Subsection~\ref{subs: BRST_DoF}.

\section{Linear gauge systems and polynomial modules}

\label{sec:gauge_and _polynom}

In this section, we look over gauge symmetries and identities through a looking glass of commutative algebra.
One can interpret the polynomial matrices, that define gauge symmetries and identities as maps of polynomial modules. This leads us to a two-sided complex of free polynomial modules~(\ref{eq:double_resolution}). 
We use this complex to prove a formula for the degree of freedom in terms of homological commutative algebra in Theorem~\ref{th: deg_free_via_Ext}.  In particular, it follows that the degree of freedom is a well-defined nonnegative integer for each linear system.  

The most important results are obtained in the case of homogeneous system. In this case we prove formulae for the degree of freedom in terms of differential orders of the gauge structure elements
(Corollaries~\ref{cor:main_homogeneous_general} and \ref{cor:main_homogeneous}) and in terms of the BRST complex Euler characteristic (Proposition~\ref{prop: DoF via BRST}). 

The necessary commutative algebra  definitions and results are collected in Appendix~\ref{app: algebra}.

\subsection{Algebraic interpretation of gauge symmetries and identities}


Consider the two-sided complex~(\ref{Gauge-for-Gauge-Sequence}).  
The linear maps ${\hat{L}_{(\bar{k})}}, \hat{T},\hat{R}, \hat{R}_{(k)}$ in this complex are interpreted as matrices.
The entries of the matrices are differential operators with constant coefficients, that is, 
polynomials of $d$ variables $\partial_0, \dots , \partial_{d-1} $. 

Let $P = \CC [\partial_0, \dots , \partial_{d-1}]$ denote the ring of differential operator polynomials with constant coefficients. 
It is the ring of complex polynomials 
of the variables $ \partial_0, \dots , \partial_{d-1}$. Since the maps  
${\hat{L}_{(\bar{k})}}, \hat{T},\hat{R}, \hat{R}_{(k)}$ are defined by matrices with entries from $P$. For example,
the map $\widehat T$ (and its matrix) defines 
a polynomial map $T^P: P^m \to P^n$. 

Moreover, one can define the gauge symmetries and gauge identities in terms of this polynomial map $T^P$.  By definition, gauge symmetries $R^i$ are defined by  the column vectors $R^i \in P^m$ such that  
$\hat{T}_{ai}\hat{R}{}^i\epsilon=0,\,\,\forall\epsilon\in\mathcal{E}$, cf.~(\ref{G-set}). This means 
that all gauge symmetries form a submodule 
$\Omega^1$ in $P^m$ which consist of all elements $r$ such that $T^P r =0$, that is,
$\Omega^1 = \ker T^P$. The ring $P$ is Noetherian, that is, any submodule of $P^m$ admits a finite generating set (by Hilbert Basis Theorem). Then any finite generating set $\{ R^i_\alpha|\alpha=1, \dots, m_1 \}$ for some $m_1\ge 0$
defines the generating set  $\{ \phi^i_\alpha = R^i_\alpha  \epsilon^\alpha |\alpha=1, \dots, m_1 \}$ for the gauge symmetries 
( as before , here $\epsilon^\alpha \in C^{\infty} (X)$ is  a generic function, cf.~(\ref{triv-set})). Vice versa, a generating set for the gauge symmetries as before defines a generating set for the submodule $\Omega^1$. Thus, the matrix $
\hat R $ is defined as a matrix with entries in $P$ such that its columns form a generating set for the submodule $\Omega^1 = \ker T^P$.

By the way, the matrix $\hat R$ defines a linear operator $R^P: P^{m_1} \to P^m$. 
In the same way as before, the matrix $\hat R_{(1)}$ can be defined as any matrix whose columns form a generating set of the submodule $\Omega^2 = \ker R^P$ in $P^{m_1}$. 
Inductively, one can define all matrices $\hat R_{(k)}$ for all positive $k$, that is, all matrices in the right-hand half of the sequence~(\ref{Gauge-for-Gauge-Sequence}).  

By construction, we get an exact sequence of polynomial modules
\begin{equation}
\label{eq:resolution_of_N}
   {P}^{n}\stackrel{T}{\longleftarrow}{P}^{m}\stackrel{R}
   {\longleftarrow}{P}^{m_1}{\leftarrow} \cdots \stackrel{{R}_{(k)}}{\longleftarrow}    \cdots
\end{equation}
This is a free resolution of the polynomial module $W = \Coker T$. By Hilbert Syzygy Theorem, one can choose the $d$-th map in the resolution such that the next term is zero. It follows that for some $k_{max} \le d-1$ the resolution ends, that is, the resolution has the form
\begin{equation}\label{eq:right_resolution}
   {P}^{n}\stackrel{T}{\longleftarrow}{P}^{m}\stackrel{R}
   {\longleftarrow}{P}^{m_1}{\leftarrow} \cdots \stackrel{{R}_{(k)}}{\longleftarrow}
   \cdots
   {\leftarrow}{P}^{m_{k_{max}}}
   \leftarrow 0\,.
\end{equation}

Now, consider the gauge identities between the field equations. By~(\ref{L-A}), a row vector $\hat{L}_A \in P^n$ form a gauge identity if and only if $\hat{L}{}^a_A \hat{T}_{ai}\phi^i\equiv 0$, that is,
$\hat{L}_A $ belongs to the kernel of the dual operator $T^*$ to the operator defined by the matrix $\hat T$ (in the dual basis, this dual operator is defined by the transpose matrix of $\hat T$).
It follows that the generating set of gauge identities is any set of gauge identities 
which generates the submodule ${\Omega^1}' =\ker T^* $ in $P^n$. Similarly to the case of gauge symmetries, this set can be chosen finite by Hilbert Basis Theorem. Then one can define the matrix $\hat L^*$ as an arbitrary matrix such that its column form a generating set for the module ${\Omega^1}'$. 

The matrix $\hat L^*$ defines a linear operator 
$L^*: P^l\to P^{n}$.
Then one can define the matrix $\hat L_{(1)}^*$ as the one whose columns form a generating set for the submodule ${\Omega^2}' = \ker \hat L^*$ in $P^l$. Inductively, one define 
$\hat L_{(k)}^*$ as the one whose columns form a generating set for the submodule  
${\Omega^k}' = \ker \hat L_{(k-1)}^*$ in $P^{k-1}$. Then, one gets a free resolution 
\begin{equation*}
   {P}^{m}\stackrel{T^*}{\longleftarrow}{P}^{n}\stackrel{L^*}
   {\longleftarrow}{P}^{l}{\leftarrow} \cdots \stackrel{{L}_{(k)}^*}{\longleftarrow}    \cdots
\end{equation*}
of the $P$-module $V  = \Coker T^*$. 
Similarly to the above, the Hilbert Syzygy Theorem guarantees that for some $l_{\Bar{k}_{max}} \le d-2$, the corresponding matrix can be chosen in such a way that the next term vanishes. Thus, 
we get a finite resolution of the module 
$V$:
\begin{equation}
\label{eq:left_resolution}
   {P}^{m}\stackrel{T^*}{\longleftarrow}{P}^{n}\stackrel{L^*}
   {\longleftarrow}{P}^{l}{\longleftarrow} \cdots \stackrel{{L}_{(k)}^*}{\longleftarrow}    \cdots
   {\longleftarrow} {P}^{l_{k_{max}}}
   \longleftarrow 0\,.
\end{equation}
Take the usual dual of the exact complex~(\ref{eq:left_resolution}), that is, consider the homomorphisms  to $P$ from each term. We obtain the following complex, which is not necessary exact:
\begin{equation}
\label{eq:dual_left_resolution}
0 
{\longleftarrow}{P}^{l_{k_{max}}}
{\longleftarrow} \cdots
\stackrel{{L}_{(k)}}{\longleftarrow}  
\cdots
 {\longleftarrow} {P}^{l}
 \stackrel{L}{\longleftarrow}
 {P}^{n}
 \stackrel{T}{\longleftarrow}
 {P}^{m}.
\end{equation}
We can unite this complex with the resolution~(\ref{eq:right_resolution}). 
Then we we obtain a complex which is a polynomial version of the complex (\ref{Gauge-for-Gauge-Sequence}):
\begin{equation}
\label{eq:double_resolution}
{\bf F}^*:
0 
{\longleftarrow}{P}^{l_{k_{max}}}
{\longleftarrow} \cdots
\stackrel{{L}_{(k)}}{\longleftarrow}  
\cdots
 {\longleftarrow} {P}^{l}
 \stackrel{L}{\longleftarrow}
 {P}^{n}
 \stackrel{T}{\longleftarrow}
 {P}^{m}
   \stackrel{R}
   {\longleftarrow}{P}^{m_1}{\leftarrow} \cdots \stackrel{{R}_{(k)}}{\longleftarrow}
   \cdots
   {\leftarrow}{P}^{m_{k_{max}}}
   \leftarrow 0\,.
\end{equation}
Recall that this complex is exact in all terms toward the right from $T$. (We will essentially use its dual complex ${\bf F}$, see~(\ref{eq: 2-sided_resolution}); this explains the star in the notation).

Note that the matrices of the differentials in this complex are transposes to the ones defining the differential $Q$ on the generators of the BRST complex~(\ref{Q}). The Fourier transform  
of the complex~(\ref{eq:double_resolution})
is isomorphic to the configuration space $\mathcal{W}$ of the BRST complex defined above by~\cite[Proposition~2]{AN}. Thus, the whole BRST complex is isomorphic to
$$ 
\bigwedge \overline{\mathcal F} \left( {\bf F}^*\right),
$$
where $\overline{\mathcal F}$ denotes the inverse Fourier transform and $\bigwedge $ 
is the notation for the free differential graded commutative algebra. Vice versa, in terms of the configuration space $\mathcal{W}$ of the BRST  complex, the complex ${\bf F}^*$ is isomorphic to the Fourier transform ${\mathcal F} \left( \mathcal{W} \right)$.

\subsection{Degree of freedom and polynomial modules}

\label{subs:N=e}

In this subsection, we prove the following useful formula 
for the degree of freedom in terms of polynomial module $N$ 
defined above. Note that there is another method to calculate the degree of freedom in terms of the involutive closure of the system, see~\cite[Ch.~8]{Seiler}.

Let 
$W$
denotes the the cokernel of the polynomial map $P^m\to P^n$ defined by the matrix $T$. 

\begin{theorem}
\label{th: deg_free_via_Ext}
The  number of physical degree of freedom is finite and is uniquely determined in terms of  the module $W$ 
by the formula
$$
{\mathcal N} = 
    e(U, d-1) \mbox{, where } U =\ext^1_P(W,P). 
$$

As a corollary, the number of physical degree of freedom is a well-defined nonnegative integer. 
\end{theorem}

The proof consists of the following  sequence of statements. 

We begin with the following two propositions which are essentially are proved in~\cite[Appendix A]{KLS2013}.

As before, suppose that 
$\Sigma$ denotes the set of all solutions of the system~(\ref{EoMs}). That is, $\Sigma$ is the kernel of the differential operator defined by the matrix $T$.
Let $T^* $ be the transpose matrix and $V$ the cokernel of the 
polynomial map defined by $T^*$, so that there is an exact sequence of (filtered) finitely generated $P$-modules
\begin{equation}
	P^n \stackrel{T^*}{\to} P^m 
	\to V \to 0. 
\end{equation}
Consider the filtration on $V$ induced by the standard filtration on $P^n$ here. 

\begin{prop}
	\label{prop:h(M)}
	In the notation above, the following equality of the Hilbert functions holds:
 $$\tilde h_V(N) = \tilde h_{\Sigma}(N)  \mbox{ for all } N\ge 0.
 $$
\end{prop}

Note that  for $n=1$, this is shown in \cite[Lemma~4, Prop.~2]{khovanski_chulkov}.

\begin{proof}
For a function $\phi$ which is analytical in a neighbourhood of $x_0$, its derivative $\partial^\mu$ at $x=x_0$ for each index $\mu\in \Zpl^d$ is
 equal to the coefficient $\phi_{\mu}$ in the power series expansion~(\ref{eq:phi_Taylor_expansion}). 
 Consider the set ${\mathcal T}$ of all consequences of the system~(\ref{EoMs}) which are polynomials with constant coefficients. The left-hand sides of all the equations are in bijections with the elements of the polynomial module $\im T^* \subset P^n$. These left-hand sides 
   are linear combinations of the derivatives $\partial^\mu$ which are equal (at $x=x_0$) to the 
unknowns $\phi_\mu$. So, the system ${\mathcal T}$ is equivalent to the system of linear equations of the variables $\phi_\mu$. Let us denote 
the space $\mathcal{C}^\infty(X)$ of smooth functions on $X$ by $\EE$
The total number of such variables having degree $|\mu|\le N$  is $\tilde h_{\EE^n}(N) = \tilde h_{P^n} (N)  $.  Then $\tilde h_{\Sigma}(N)$
is equal to the maximal number of the variables of this degree which are linearly independent modulo the equations. The number of the equations which involve only the variables $\phi_\mu$ of degree at most $q$ is 
equal to $\tilde \tilde h_{\im T^*}(N) $. Thus,
$$
\tilde h_{\mathcal{K}} (N) = h_{\EE^n}(N) - \tilde h_{\im T^*}(N)  = 
\tilde h_{P^n}(q)  - \tilde h_{\im T^*}(N)  =  
\tilde h_{V}(N).
$$
\end{proof}

Consider the initial segments of the exact sequence~(\ref{eq:right_resolution}):
$$
   {P}^{n}\stackrel{T}{\longleftarrow}{P}^{m}\stackrel{R}
   {\longleftarrow}{P}^{m_1}
   .
   $$
   Its dual has the form
\begin{equation}
	\label{eq: C' to S'}
	P^n \stackrel{T^*}{\to} P^m \stackrel{R^*}{\to} P^{m_1}
 .
\end{equation}
Let $Q = \ker R^* \subset P^m$, so that the homology module in the central term here is the module $U = Q/\im {T^*}$. 
We consider the Hilbert series of this module 
for 
the filtration induced by the standard filtration of  $P^m$.

\begin{prop}
	\label{prop:phys_solutions}
 In the notation of Subsection~\ref{subs:deg_freedom_definition}, 
 there is an equality of Hilbert functions
 $\tilde h_U(N) = \tilde h_{\mathcal{K}}(N) - \tilde h_{\mathcal{K'}}(N)$ for all $N\ge 0$.
\end{prop}

\begin{proof}
	In the notation of  Subsection~\ref{subs:deg_freedom_definition}, 
 consider the factor module 
 $\mathcal{K}/\mathcal{K'}$.

	The module $\Sigma'$ consists of all elements of the form $s \phi$, where $s$ runs the module $\im R \subset P^m$ and $\phi \in \EE^m$. 
	So, a jet $j \in \pi_q(\mathcal{K})$ 
 belongs to $\pi_q(\mathcal{K'})$ 
 iff it satisfies some equation of the form 
	$sj =0$, where $s \in \im R $  has degree at most $q$. The linear space of such equations has dimension  $\tilde h_{\im R^*}(N)$ (where the filtration on the module $\im R^*$ is induced by the one on $P^m$). So, the dimension of the quotient module $\dim	\pi_N(\mathcal{K}/\mathcal{K'})=\tilde h_{\mathcal{K}/\mathcal{K'}}(N)$
	is equal to $\tilde h_{\mathcal{K}}(N) - \tilde h_{\im R^*}(N)$.
	
	 Recall that $\tilde h_{\mathcal{K}}(N) = \tilde h_V(N)$ 
	by Proposition~\ref{prop:h(M)}.  
	Moreover, $\tilde h_U(N) = \tilde h_Q(N) - \tilde h_{\im R^*}(N) =  \tilde h_Q(N) - \tilde h_{P^m}(N)+ 
	\tilde h_V(N)$ (where all filtrations are induced by the one on $P^m$). From the exact sequence 
	$$
	0\to Q \to P^m \to \im R^* \to 0
	$$
 of $P$-modules
	we deduce the equality $ \tilde h_{\im R^*}(N) =  \tilde h_{P^m}(N)- \tilde h_{Q}(N)$. Finally, we obtain
	$$
	\tilde h_{\mathcal{K}}(N) -\tilde  h_{\mathcal{K'}}(N)
  = \tilde h_{\mathcal{K}/\mathcal{K'}}(N)  = \tilde  h_{\mathcal{K}}(N) - \tilde h_{\im R^*}(N) = 
	\tilde h_M(N) - \tilde h_{P^m}(N)+\tilde h_Q(N) = \tilde h_Q(N).
	$$
\end{proof}
	
	\begin{lemma}
		\label{lem: deg_free}
Let $U$ denotes the $P$-module $\ker R^*/\im T^*$ as above. Then the 
physical degree of freedom 
is equal to
$$
{\mathcal N} =  \left\{ \begin{array}{ll}
	0, & \dim U \le d-2, \\
    e(U, d-1)
	\ne 0, & \dim U=d-1,\\
	\infty, & \dim U = d. 
\end{array}
\right.
$$
	\end{lemma}
	
	\begin{proof}
By the definition (\ref{eq:deg_free_definition}),
$$
{\mathcal N} = \frac{1}{d-1}\lim_{p\to \infty} p \frac{h_{\Sigma} (p) - h_{\Sigma'} (p)}{\binom{p+d-1}{p}} = 
\lim_{p\to \infty} \frac{p h(U,p) }{(d-1)\binom{p+d-1}{p}},
$$		
where 
$\binom{p+d-1}{p} = h(\EE,p) = \frac{1}{(d-1)!}p^{d-1}+o(p^{d-1})$.
If $\dim U=d$, then $h(U, p) = \frac{e(U,d)}{(d-1)!}p^{d-1}+o(p^{d-1})$ by Proposition~\ref{prop:dim_n_mult}, so that 
$$
{\mathcal N} = \lim_{p\to \infty}
\frac{p \cdot e(U,d)}{d-1}
= \infty.
$$
Now, suppose that $\dim U\le d-1$. According to Proposition~\ref{prop:dim_n_mult}, we have
$h(U, p) = \frac{e(U,d-1)}{(d-2)!}p^{d-2}+o(p^{d-2})$.
Thus, 
$$
{\mathcal N} = \lim_{p\to \infty}
\frac{p}{d-1} \frac{e(U,d-1)p^{d-2}/(d-2)!}{
	p^{d-1}/(d-1)!} = 
e(U,d-1).
$$ 
In remains to note that $e(U,d-1) = 0$ if $\dim U\le d-2$.  
	\end{proof}

By definition, the module $U$ here is uniquely determined by the matrix $T$ only (and does not depend on the choice of $R$ and $L$). 
To describe it in homological terms,
consider the free resolution~(\ref{eq:resolution_of_N}) of the cokernel $W$ of the polynomial map $T$.
 Taking the homomorphism of this resolution to $P$, we obtain a complex whose homology calculates $\ext^\cdot_P(W,P)$. This complex is isomorphic to the left-hand part of the complex~(\ref{eq:double_resolution}):
 \begin{equation*}    
0 
{\longleftarrow}{P}^{l_{k_{max}}}
{\longleftarrow} \cdots
\stackrel{{L}_{(k)}}{\longleftarrow}  
\cdots
 {\longleftarrow} {P}^{l}
 \stackrel{L}{\longleftarrow}
 {P}^{n}
 \stackrel{T}{\longleftarrow}
 {P}^{m}
 \end{equation*} 
On the other hand, the homology in the term $P^n$ here is, by definition, the module $U$. We have obtained 
 
 \begin{prop}
 	\label{prop:U = ext}
 	In the above notation, the module $U$ is isomorphic to $\ext^1_P(W,P)$. 
 \end{prop}
 
 Now, one can use Proposition~\ref{prop: groth} to evaluate the dimension of the module $U$.  
The next corollary proves Theorem~\ref{th: deg_free_via_Ext}.

\begin{cor}
	For each module $W$,  the  number of physical degree of freedom is equal to $e(U, d-1)$ and is finite, 
	$${\mathcal N} = e(U, d-1)<\infty.
	$$
\end{cor}

\begin{proof}
	By Proposition~\ref{prop: groth}, $\dim U \le d-1$.
	It follows from Lemma~\ref{lem: deg_free} that 
	${\mathcal N} = e(U, d-1)<\infty$.
\end{proof}

\subsection{Degree of freedom for homogeneous systems}

\label{subs: homogen}

Let us introduce a $Z$-grading for the fields and the gauge structure induced by the differential order. Recall that we assign an integer differential order $\theta_i$ to each field variable
$\phi^i$ ($i=1, \dots, m$) in~(\ref{EoMs}).

From now, we assume that the system (\ref{EoMs}) is homogeneous, that is, for each equation $T_\alpha =0$ all summands in  (\ref{T-hat}) 
have the same total differential order $k_a = \ord T_a$.  
In most cases, 
we put
$\theta_i = 0$, so we usually consider the fields as dimensionless (we refer to the induced  collection of differential orders as {\em standard grading}). In the case of standard grading, one gets $k_a = k_{ai}^{max}$ for each $i$, see~(\ref{T-hat}), (\ref{OrderT}).

The 
choice of the differential orders $\theta_i$  
 unambiguously 
defines the differential order for all  ghosts and antifields.
For homogeneous system, one can choose the ghost and antifield variables in such a way that each partial differentiation $\partial_{\mu}$ has is homogeneous of differential order $1$.
So, the differential order corresponds to the standard grading of the differential polynomial ring $P$ by the degrees of polynomials.\footnote{More generally, one can assign a positive integer degree $\delta_{\mu}$ to each $\partial_{\mu}$ ($\mu=0, \dots, d-1$), so that the operator  $\partial_{\mu}$ is then considered as a differential operator of order $\delta_{\mu}$. Then $P$ should be considered as a weighted polynomial ring.  This grading can be thought of as if we assign the dimension of length to every coordinate $x^\mu$, then the opposite dimension, denoted $\delta_{\mu}$, is assigned to $\partial_\mu$.
The choice of the dimensions $\delta_{\mu}$ and $\theta_i$ unambiguously defines dimensions of all the ghosts and antifields. As we see, the final result for the DoF number does not depend on the choice of the dimensions for the fields.} 
Generally, the grading by differnetial order is not equal to one defined by the  ghost numbers.

For homogeneous systems, one can calculate explicitly the degree of freedom in terms of some vector space dimensions of graded components. Let us discuss the details.

So, we assume that the system is homogeneous. 
Then $P$-modules $W = \Coker T$, $V = \Coker T^*$, and $U =\ext^1_P(W,P)$
introduced above are graded. 
Moreover, then the graded $P$-module  $V$ admits a graded free resolution of the form~(\ref{eq: free_resol_for_M}).  We denote the terms of the resolution by $F_i^V$, $i=0,1,\dots$
Each term is a finitely generated free graded module,  $F_i^V = \bigoplus_j P(-j)^{b_{ij}}$ (where the Betti numbers $b_{ij}$ of this resolutions  are some nonnegative integers). 
One can choose 
$F_0^V$ and  $F_1^M$ to be isomorphic to $P^m$ and $P^n$ as non-graded modules, so that the map $d_1^V: F_1^V \to F_0^V$ is defined by the matrix $T^*$.
If $V$ is generated in degree 0 (so that all $\theta_i$ are zero), here $F_0^V = P^m$ with standard grading.
Generally, we have $F_0^V = \bigoplus_i P(-\theta_i )$ and $F_1^V = \bigoplus_a P(-k_a )$. 
Note that here one can assume $F_i^V = 0$ for $i>d$ since the homological dimension of the ring $P$ is $d$.
So, the resolution gives an exact sequence of graded $P$-modules
\begin{equation}
	\label{eq: our M resol}
0\to F_d^V	\to\dots \to  F_1^V \stackrel{T^*}{\to} F_0^V  
	\to V \to 0. 
\end{equation}
Recall the notations $Q^i_V(z) = Q_{F_i^V}(z)$ and 
$Q_V(z) = \sum_i (-1)^i Q^i_V(z)$. 


Since the map $P^n \stackrel{T^*}{\to} P^m$ is homogeneous with the grading induced by the isomorphisms of (non-graded) modules $P^n \simeq F_1^V  = \bigoplus_{j} P(-j)^{b_{1j}}$
and $P^m \simeq F_0^V  = \bigoplus_{j} P(-j)^{b_{0j}}$, its dual map 
$P^m \stackrel{T}{\to} P^n$ is also homogeneous with the grading induced by the 
isomorphism $P^n \simeq (F_1^V)^* = \bigoplus_{j} P(j)^{b_{1j}}$ and 
$P^m \simeq (F_1^V)^* = \bigoplus_{j} P(j)^{b_{1j}}$. So, in the exact sequence~(\ref{eq:double_resolution}) all modules may be considered as graded. Hence there exists a  graded resolution of the graded module $W$
\begin{equation}
	\label{eq: resol for W}
0\to F_d^W \to \dots 
\to	
F_2^W
\stackrel{R}{\to} 
F_1^W 
\stackrel{T}{\to} 
F_0^W 
	\to W \to 0 
\end{equation}
with $F_0^W = \bigoplus_{j} P(j)^{b_{1j}} \simeq P^n, F_1^W = \bigoplus_{j} P(j)^{b_{0j}} \simeq  P^m$ (where $F_1^W$ 
is isomorphic to $ P^m$ with standard grading if $V$ is generated in degree 0).  Again, we put $Q^i_W(z) = Q_{F_i^W}(z)$ and obtain  $Q_W(z) = \sum_i (-1)^i Q^i_W(z)$.
	
If we apply the functor $\Hom_P(-, P)$ to the resolution~(\ref{eq: resol for W}), we obtain the dual complex 
\begin{equation}
\label{eq: dual resol for N}
0\to W^* \to 
(F_0^W)^* 
\stackrel{T^*}{\to} 
(F_1^W)^* \to \dots 
\to 
(F_d^W)^* \to 0,
\end{equation}
where $(F_0^W)^*  = \bigoplus_{j} P(-j)^{b_{1j}} = F_1^V$
and $(F_1^W)^* = \bigoplus_{j} P(-j)^{b_{0j}} =  F_0^V$. 

If we glue the dual resolution~(\ref{eq: dual resol for N}) for $W$ with the resolution~(\ref{eq: our M resol}) for $V$, we get a two-sided complex of free modules
\begin{equation}
	\label{eq: 2-sided_resolution}
	{\mathbf F}:	0\to F_d^V	\to\dots \to  F_1^V \stackrel{T^*}{\to} F_0^V \stackrel{R^*}{\to} (F_2^W)^* \to \dots \to (F_d^W)^* \to 0.
\end{equation}
This complex is dual to~(\ref{eq:double_resolution}).
The only nonzero homologies of this complex may occur in the rightmost terms  $F_0^V = (F_1^W)^*, (F_2^W)^*, \dots , (F_d^W)^*  $; the homology modules are isomorphic respectively to $\ext^i_P(W,P)$ for $i=1, \dots, d$. 

Let $ Q_{\bf F} (z)$ be the polynomial obtained from the  Euler characteristic of the two-sided complex~(\ref{eq: 2-sided_resolution}),
$$
Q_{\bf F}(z) = \sum_{i=0}^d (-1)^i Q_{F_i^V }(z) - \sum_{i=2}^d (-1)^i  Q_{(F_i^W)*}(z) = 
Q_V(z) - \sum_{i=2}^d (-1)^i Q_W^i(z^{-1}).
$$

Now, let us consider a special case. Suppose that the system~(\ref{EoMs}) does not admit gauge symmetry transformations, $m_1=0$. This means that all terms in the free resolution~(\ref{eq: resol for W}) beginning with $F_2^W$ vanish, that is,   $\pd W \le 1$.
Then the second sum here is zero, so that we have 
$Q_{\bf F}(z) = \sum_{i=0}^d (-1)^i Q_{F_i^V }(z)= Q_V(z)$. On the other hand, in this case
$U = \ker R^*/\im T^* = (F_1^W)^*/\im T^* = F_0^V/\im T^* = V$. Then we obtain

\begin{cor}
	\label{cor: hom_pd<=1}
	Suppose that the system~(\ref{EoMs}) does not admit gauge symmetry transformations, or, equivalently, $\pd W \le 1$. Then $Q_U(z) = Q_{\bf F}(z)$. As a corollary,  
 ${\mathcal N} =
	- Q_{\bf F}'(1)$.
\end{cor}

{\em Remark.} It follows that for a system without gauge  symmetry transformations,
${\mathcal N} = e(V,d-1) = -Q_V'(1)
	= \sum_j j \sum_{i=1}^d (-1)^{i+1} { b_{ij}^V }$.

If $\pd W >1$, the polynomial  equality  $Q_U(z) =  Q_{\bf F}(z)$ does not generally hold. For example, suppose that
$\pd W =2$. (A simple example of the matrix $T$ such that $\pd W=2$  is the matrix $T  = (\partial_0, \partial_1 )$ for $d=2$, $n=1$, $m=2$.)
Then $\ext_P^2(W, P) \ne 0$ (by the graded version of the Grothendieck non-vanishing theorem), and all higher  $\ext_P^i(W, P)$ vanish. 
Since the  module $\ext_P^2(W, P)$ is the homology of the complex 
${\bf F}$ at the term $(F_2^W)^*$, we have 
$Q_{\bf F}(z) =  Q_U(z) - Q_{\ext_P^2(W, P)}(z) \ne Q_U(z)$.  

Still, the equality that connects the degree of freedom and the multiplicity holds in the general case as well. 

\begin{theorem}
	\label{th:deg_free n graded 2-sided}
	If the system~(\ref{EoMs}) is homogeneous, 
 the physical degree of freedom satisfies 
	$${\mathcal N} = - Q_{\bf F}'(1) = 
	\sum_j j 
	\sum_{i=1}^d (-1)^{i+1}  
	\left( 
	b_{ij}^V
	+ b_{ij}^W \right) 
	. 
	$$
\end{theorem}

\begin{proof}
	Denote $E_i = \ext_P^i(W, P)$ for each $i\ge 1$.
	By Proposition~\ref{prop: groth}, we have $\dim  E_i \le d-i$. Hence the polynomial $Q_{E_i}(z)$ is divisible by $(z-1)^{i}$, see~(\ref{eq:Q(M)=(1-z)^(d-D)p_M}). In particular, $Q_{E_i}(1)' =0$
	if $i\ge 2$.
	
	The Euler characteristic for the homology of the complex ${\bf F}$ is, by definition,
	$$
	\chi_{H{\bf F}} (z) = H_P(z) \sum_{i=1}^d (-1)^{i+1} 
	Q_{E_i}(z).  
	$$
	Since $Q_{\bf F}(z) = \chi_{\bf F}(z)H_P(z)^{-1} = \chi_{H{\bf F}} (z) H_P(z)^{-1} $, we have   $$
	Q_{\bf F}'(1) =  \sum_{i=1}^d (-1)^{i+1} 
	Q_{E_i}'(1) =  Q_{E_1}'(1) = Q_{U}'(1) = -
	{\mathcal N} .
	$$
\end{proof}

The expression of DoF in terms of the gauage structure is given in the next corollary.

\begin{cor}
\label{cor:main_homogeneous_general}
Suppose that the system~(\ref{EoMs}) is homogeneous and the grading is not necessarily standard, and suppose that all maps $R_{(k)}, L_{(k)}$ are chosen to be homogeneous as well.
Then the  physical degree of freedom can be calculated as follows:
	$${\mathcal N} = \sum_{a} k_a
 -\sum_{i} \theta_i-\sum_{k= 0}^{\Bar {k}_{max}} (-1)^k \sum_{A_k} l^{(k)}_{A_k} - 
 \sum_{k= 0}^{{k}_{max}} (-1)^k \sum_{\alpha_k}  r_{\alpha_k}^{(k)}
 $$
 $$
 = \sum_{a} k_a
 -\sum_{i} \theta_i - \sum_{k= 0}^{d-2}(-1)^k \left(
 l^{(k)}_{A_k} + r_{\alpha_k}^{(k)}
 \right)
	. 
	$$
    \end{cor}

\begin{proof}
The maps in the resolutions for the modules $V$ and $W$
are in correspondence with the complex~(\ref{Gauge-for-Gauge-Sequence}) defined for  gauge symmetries and gauge identities. 
The comparison of the complexes~(\ref{Gauge-for-Gauge-Sequence}) and~(\ref{eq: 2-sided_resolution}) gives the following equalities for the Betti numbers.
Here $\beta_{i}^V$ (respectively, $\beta_{i}^W$)
denotes the sum of all generators' degrees for the free module $F_i^V$ (resp., $F_i^W$), that is, the sum $\sum_j j b_{ij}^V$ or, respectively, $\sum_j j b_{ij}^W$. 
$$
\begin{array}{l}
     \beta_{0}^V = \sum_j \theta_j = -\beta_{1}^W,\\
     \beta_{1}^V = \sum_j k_j = - \beta_{0}^W,\\
     \beta_{i}^V = \sum_j r^{(i-2)}_j, i\ge 2,\\
     \beta_{i}^W = \sum_j  l^{(i-2)}_j, i\ge 2.\\
   \end{array}
$$
Then the corollary follows from  Theorem~\ref{th:deg_free n graded 2-sided}, since 
we have
$${\mathcal N} = 
	\sum_j j 
	\left( 
	\sum_{i=1}^d (-1)^{i+1}  
	\left( 
	b_{ij}^V
	+ b_{ij}^W \right) 
	\right)
 = \beta_{1}^W + \beta_{1}^V - 
 \sum_{i=2}^d (-1)^i (\beta_{i}^V +
     \beta_{i}^W )
     $$
     $$
     = 
     \sum_j \left(
     -\theta_j +k_j
 -  \sum_{i=2}^d  (-1)^i  (r^{(i-2)}_j+ l^{(i-2)}_j)
     \right) =      
     -\sum_i \theta_i + \sum_a k_a
 -  \sum_{k= 0}^{d-2}(-1)^k \left(
 l^{(k)}_{A_k} + r_{\alpha_k}^{(k)}
 \right)
	. 
	$$
\end{proof}

In the particular case of the standard grading, we put $\theta_i=0$ for all $i$ and obtain the second equality of Theorem~\ref{th:main_homogeneous} as follows.

\begin{cor}
\label{cor:main_homogeneous}
Suppose that the system~(\ref{EoMs}) is homogeneous with standard grading, and all maps $R_{(k)}, L_{(k)}$ are chosen to be homogeneous as well.
Then the  physical degree of freedom is equal to the sum
	$${\mathcal N} = \sum_{a} k_a-\sum_{k= 0}^{\Bar {k}_{max}} (-1)^k \sum_{A_k} l^{(k)}_{A_k} - 
 \sum_{k= 0}^{{k}_{max}} (-1)^k \sum_{\alpha_k}  r_{\alpha_k}^{(k)}
	. 
	$$
    \end{cor}

Let ${\mathcal N}_T = {\mathcal N}$ denotes the physical degree of freedom for the system~(\ref{EoMs}) defined by the matrix $T$, and let  ${\mathcal N}^\dagger = {\mathcal N}_{T^\dagger}$
denotes the same number for the conjugate system which is defined by the Hermitian transpose matrix $T^\dagger$. 
\begin{cor}
	\label{cor: conjugate N = N }
	If the matrix $T$ is homogeneous, then ${\mathcal N}^\dagger = {\mathcal N}$. 
\end{cor}

\begin{proof}
	Consider the  modules $W$ defined by the matrix $T^\dagger$ in place of $T$, that is, $W^\dagger  = \Coker T^\dagger$
	and $V^\dagger =\Coker (T^\dagger)^*$. 
The automorphism of the ring $P$ defined by the conjugation of complex numbers and the change of the signs of the variables $\partial_i $ ($i=0, \dots, d-1$) induces an autoequivalence of the category of $P$-graded modules; denote this autoequivalence by ${\mathcal L}$. Note that the isomorphism 
maps the matrix $T$ to $T(^\dagger)^*$, so 
that ${\mathcal L}W =  V^\dagger$ and ${\mathcal L}V =  W^\dagger$.  Obviously, this autoequivalence preserves (free) graded modules, so that it maps the  resolutions~(\ref{eq: our M resol}) and~(\ref{eq:resolution_of_N}) to the analogous resolutions for the modules ${\mathcal L}V$
${\mathcal L}W$ as follows: 
\begin{equation}
	\label{eq: resol for V_dagger}
0\to F_d^W \to \dots 
\to	
F_2^W
\to
F_1^W 
\stackrel{(T^\dagger)^*}{\to} 
F_0^W 
	\to {\mathcal L} W (=V^\dagger) \to 0, 
\end{equation}
\begin{equation}
	\label{eq: resol for W_dagger}
0\to F_d^V \to \dots 
\to	
F_2^V
{\to} 
F_1^V 
\stackrel{T^\dagger}{\to} 
F_0^V 
	\to {\mathcal L} V (=W^\dagger)\to 0 
\end{equation}
If we glue the resolution~(\ref{eq: resol for V_dagger}) for $V^\dagger$ with the dual to the resolution~(\ref{eq: resol for W_dagger}) for $W^\dagger$, we get a two-sided complex of free modules
\begin{equation}
	\label{eq: 2-sided_resolution_dagger}
	{\mathbf F^\dagger}:	0\to F_d^W	\to\dots \to  F_1^W \stackrel{(T^\dagger)^*}{\to} F_0^W {\to} (F_2^V)^* \to \dots \to (F_d^V)^* \to 0.
\end{equation}
This is the complex suitable for the degree of freedom calculation in the case of the Hermitian transpose matrix $T^\dagger$. The degree of freedom does not depend of the grading, so that we may assume that the grading on the module $\Coker T^\dagger$ is such all terms of this complex are dual to the corresponding terms of the complex~(\ref{eq: 2-sided_resolution}). Then we have  
$$	
	{\mathcal N}^\dagger = 
 - Q_{\bf F^\dagger}'(1) = \sum_j j 
	\left( 
	\sum_{i=1}^d (-1)^{i+1}  
	\left( 
	b_{ij}^W
	+ b_{ij}^V \right) 
	\right) = 
 {\mathcal N}.
 $$
 \end{proof}

\begin{exam}[DoF count in the system of Maxwell equations for the strength tensor] 
\label{ex: maxwell}
Consider the electromagnetic field strength tensor  $F^{\mu\nu}$ in $d=4$ Minkowski space. The tensor is antisymmetric $F^{\mu\nu}=- F^{\nu\mu}$, and the labels can be lowered by Minkowski metrics turning it not 2-form. The Maxwell equations for the free field mean that the two form is closed and co-closed. In the components the equations read
\begin{equation}\label{Maxwell-eq}
  T^\mu_1\equiv\partial_\nu F^{\mu\nu}=0\, , \quad T^\mu_2\equiv\epsilon^{\mu\nu\lambda\rho}\partial_\nu F_{\lambda\rho}=0 \, .
\end{equation}
There are eight homogeneous equations of the first order, $n_1=8$. Among these equations there are two identities of the second order, $l_2=2$
\begin{equation}\label{I-Maxwell}
    \partial_\mu T^\mu_{1,2}\equiv 0 \, .
\end{equation}
these identities are irreducible. The equations are not gauge invariant.

The above identities correspond to the complex 
\begin{equation*}
0\to P^2(-2) \to P^8(-1) \stackrel{T^*}{\to} P^6 \end{equation*}
Direct calculations
show that this is indeed a resolution of the form~(\ref{eq: our M resol}) for the module $V= \Coker T$. 
Moreover, an analogous calculation shows that there is no non-trivial gauge transformation, that is, the module $W = \Coker T^*$ admits the shortest possible free resolution
$$
0\to  P^6 \stackrel{T}{\to} P^8(1) 
$$
It follows that the two-sided complex defined in (\ref{eq: 2-sided_resolution}) obtains the form
\begin{equation}
\label{eq: F for Maxwell}
	{\mathbf F}:	
    0\to P^2(-2) \to P^8(-1) \stackrel{T^*}{\to} P^6 
    \to 0.
\end{equation}
Then the count by Corollary~\ref{cor:main_homogeneous}  reads:
$$
  N_{DoF}=-8\cdot (-1)+2\cdot (-2)= 4 \, .
$$
It is the correct degree of freedom for the $e/m$ field by the phase space count. 
\end{exam}

\subsection{Degree of freedom from the perspective of the BRST complex}

\label{subs: BRST_DoF}

Now, let us connect the degree of freedom with the BRST complex (cf. Subsection~\ref{subs: BRST}). We will show that DoF can be expressed in terms of the Euler characteristic of this complex. Following the previous subsection, we assume here that the matrix $T$ is homogeneous.

Our first objective is to define the Euler characteristic of BRST complex. This definition may be of independent interest. Whereas the direct definition of Euler characteristic in this context leads to divergent infinite sums, we get around this obstruction  by uniformly assigning sufficiently large dimensions to all fields. 

This 
increase the formal differential orders of all BRST generatings, that is, fields, anti-fields, and higher-order gauge symmetries and gauge identities. As a result, the complex acquires a well-defined Euler characteristic, expressed as a formal power series in a variable 
 $z$. 
 
 Although the resulting Euler characteristic depends on the choice of field dimensions, the residue of its logarithmic derivative at infinity remains dimension-independent. This residue gives the degree of freedom, as shown in Proposition~\ref{prop: DoF via BRST} below.
 
Recall from Subsection~\ref{subs: BRST} that the 
 configuration space $\mathcal{W}$ of the BRST complex
 is the direct sum of the free $\mathcal{C}^\infty(X)$-modules $\mathcal{V}_i$ and $ \Bar{\mathcal{V}}_j$ ($i=0, \dots, k_{max}$, $j=0, \dots, \Bar k_{max}$)
 corresponding to fields (or gauge transformation parameters) and antifields (of gauge identities). 
 Moreover, the complex $\mathcal{W}$ is isomorphic to the inverse Fourier transform of the complex $ {\bf F}^*$ dual to the complex $ {\bf F}$ defined in~(\ref{eq: 2-sided_resolution}). Therefore, one can assign to the generatings of the $i$-th
 term  $\mathcal{W}_i$ of the complex $\mathcal{W}$ the same orders as the ones of the basis elements of the dual module to the $(-i)$-th term of the complex $ {\bf F}^*$, that is, the same orders as the ones for the basis elements of the $i$-th term of the complex ${\bf F}$. The number of such generatings of 
 order $j$ is equal to the bigraded Betti number $b_{ij} = b_{ij}^{\bf F}$; by definition, these are $b_{ij} = b_{ij}^V$ for $i\ge 0$ and $b_{ij} = b_{1-i,-j}^W$ for $i<0$.

Since the whole BRST complex is a free commutative differential graded algebra  $ 
C = \bigwedge \mathcal{W}
$ over the ring $\mathcal{E} = \mathcal{C}^\infty(X) $, one can define its bigraded Poincar\'e series 
\begin{equation}
\label{eq: Poincare_for_BRST}
P_C(t,z) = \sum_{i,j \in \ZZ} t^i z^j \rk_\mathcal{E} C_{ij}  = 
\prod_{i,j\in \ZZ} \left( 1-(-1)^i t^iz^j \right)^{(-1)^{i+1} b_{ij}}    
\end{equation}
and Euler characteristic
\begin{equation}
\label{eq: Euler_for_BRST}
\chi_C(z) = P_C(-1,z)= \sum_{i,j \in \ZZ} (-1)^i z^j \rk_\mathcal{E} C_{ij}  = 
\prod_{i,j\in \ZZ} \left( 1-z^j \right)^{(-1)^{i+1} b_{ij}}.
\end{equation}

Note that the above Poincar\'e series and Euler characteristic could be  ill-defined if the orders of the BRST complex generators are not all positive or negative. For example, if there 
exists a generators $z\in C_{00}$ (of zero order), then the rank $\rk_\mathcal{E} C_{00}$ is infinite, so that the sums in the definitions of the Poincar\'e series and the Euler characteristic are undefined. 

On the other hand, 
any uniform shift of the formal differential orders
 $\theta_i$ for the field variables  $\phi^i$
 induces  corresponding order shifts  throughout the entire gauge structure. 
 To address this, we apply
 Let us apply  a simultaneous integer shift $\theta_i \mapsto \theta_i+c$  to all fields variable orders $\theta_i$, where $c > \max_{k,\alpha_k} \rho_{\alpha_k}^{(k)}$. The same shift is then applied to all elements of the gauge structure.
 
After the shift, all generatings of the BRST complex acquire positive orders, ensuring that the formal power series  $
P_C(t,z)$ and $\chi_C(z)$ will be well-defined. 

To connect the Euler characteristic introduced above with homology, 
consider an algebraic dual version of the BRST complex $C$, namely, the  differential graded $P$-algebra 
$B = \bigwedge  {\bf F}^* = \bigwedge  ( {\mathcal F} \mathcal{W} )$, which is the Fourier dual to $C$. 
The algebra $B$ is bigraded as a vector space, $B = \bigoplus_{i,j\in \ZZ} B_{ij}$; the first grading is induced by the ghost number and the second one is induced by the grading of $P$-modules. 
Its bigraded Poincar\'e series is 
$$
P_B(t,z) = \sum_{i,j \in \ZZ} t^i z^j \dim B_{ij}  = H_P(z) 
\prod_{i,j\in \ZZ} \left( 1-(-1)^i t^iz^j \right)^{(-1)^{i+1} b^{{\bf F}^*}_{ij}},
$$
where the Betti numbers of the complex ${\bf F}^*$ are 
$b^{{\bf F}^*}_{ij} = b_{-i,-j}$ and $H_P(z) = (1-z)^{-d}$.
So, its Euler characteristic is 
\begin{equation}
\label{eq:Euler_BRST}
\begin{split}
\chi_B(z) = P_B(-1,z)= H_P(z) \sum_{i,j \in \ZZ} (-1)^i z^j \dim B_{ij}  = 
H_P(z) \prod_{i,j\in \ZZ} \left( 1-z^j \right)^{(-1)^{i+1} b_{-i,-j}}
\\
 = H_P(z) \prod_{i,j\in \ZZ} \left( 1-z^{-j} \right)^{(-1)^{i+1} b_{i,j}}
= H_P(z) \chi_C(z^{-1}).
    \end{split}
\end{equation}
If $\homol B$ is the bigraded homology algebra of the 
differnetial graded algebra $B$, its Euler characteristic is equal to the one of $B$, 
$$
\chi_{\homol B}(z) = \sum_{i,j \in \ZZ} (-1)^i z^j \dim \homol B_{ij} = \chi_B(z).
$$

\begin{prop}
\label{prop: DoF via BRST}
Suppose that the matrix $T$ is homogeneous and
the formal differential orders  $\theta_i$ of the fields $\phi^i$  are chosen in such a way that all 
generators of the BRST complex have positive orders, so that its Euler characteristic $\chi_C(z)$ is a well-defined formal power series.
Then the degree of freedom can be uniquely recovered from $\chi_C(z)$ as follows:
$$
{\mathcal N} = \lim_{|z|\to \infty} z \left( \ln \chi_C(z) \right) '.
$$
So, the degree of freedom is equal to minus the residue at infinity 
of the logarithmic derivative $l(z) = \left( \ln \chi_C(z) \right)'
= \chi_C(z)' / \chi_C(z) 
$
of $\chi_C(z)$.

In terms of the Euler characteristic $\chi_B(z)$ of the algebraic dual BRST complex $B$, the degree of freedom is equal to minus the residue of its logarithmic derivative at 0, 
$$
{\mathcal N} = - \lim_{z\to 0} z \left( \ln \chi_B(z) \right) '.
$$
\end{prop}

\begin{proof}
From the above formula 
$$ \chi_C(z) = 
\prod_{i,j\in \ZZ} \left( 1-z^j \right)^{(-1)^{i+1} b_{ij}}
$$
we deduce that the logarithmic derivative $l(z)$ is a rational function of the form 
$$
l(z) = \left( \ln \chi_C(z) \right) ' =  -\sum_{i,j\in \ZZ} (-1)^{i+1} b_{ij} \frac{j z^{j-1}}{1-z^j} = 
\sum_{j\in \ZZ}  q_j \frac{z^{j-1}}{1-z^j}, 
$$
where $q_j = j \sum_{i\in \ZZ} (-1)^{i} b_{ij} $ are integers. Since all generators of the BRST complex have positive degrees, we have $b_{ij}=0$ and $q_j = 0$ for all $j\le 0$. Then  
$$
zl(z) = \sum_{j>0}  q_j  \frac{z^j}{1-z^j} 
= \sum_{j>0}  q_j \left( \frac{1}{1-z^j} - 1 \right) = \sum_{j>0}  q_j \frac{1}{1-z^j}
-\sum_{j>0}  q_j.
$$
For $|z| \to \infty$, this formula gives  
$$
\lim_{|z|\to \infty} zl(z) =  -\sum_{j>0}  q_j
= -\sum_{i,j\in \ZZ} j (-1)^{i} b_{ij} 
 = - Q_{\bf F}'(1) ={\mathcal N}.
$$
This proves the first equality.

Now, let us express the same value in terms of the Euler characteristic of the algebraic dual BRST complex $B$. According to~(\ref{eq:Euler_BRST}), 
we re-write the above formula as follows:
$$
{\mathcal N} = 
\lim_{|z|\to \infty} z \left( \ln \chi_C(z) \right) ' 
= - \lim_{|z|\to \infty} z
\left( \ln \frac{\chi_C(z^{-1})}{H_P(z^{-1})} \right)' ,
$$
where $H_P(z^{-1}) = (1-z^{-1})^{-d}$. Then
$$
{\mathcal N} =
- \lim_{|z|\to \infty} z
\left( \ln \left(\chi_C(z^{-1}) (1-z^{-1})^{d} \right) \right)'
= - \lim_{|z|\to \infty} z \left( \ln (\chi_C(z^{-1})\right) '
- \lim_{|z|\to \infty} z\cdot
\left( d \cdot \ln (1-z^{-1})'\right)' ,
$$
where the second summand is equal to
$$
- \lim_{|z|\to \infty} z\cdot
\left( d \cdot \ln (1-z^{-1})\right)'
= -
\lim_{|z|\to \infty} z \cdot
d \cdot \frac{z^{-2}}{1-z^{-1}}
= d \lim_{|z|\to \infty} \frac{z^{-1}}{1-z^{-1}}
 = 0.
 $$

Introduce a new variable $\zeta = z^{-1}$. Using the equality $\frac{d \zeta}{d z} = -z^{-2} = -\zeta^2$, we obtain the following formula for ${\mathcal N}$ (in the two rightmost parts the primes denote the derivatives with respect to $\zeta$):
$${\mathcal N}
 =
 - \lim_{|z|\to \infty} z \left( \ln (\chi_B(z^{-1})\right)'= 
 - \lim_{\zeta \to 0} \zeta^{-1} \left( \ln \chi_B(\zeta) \right)' (-\zeta^2)
= 
\lim_{\zeta \to 0} \zeta \left( \ln \chi_B(\zeta) \right)'.
$$
Up to variable renaming, this is the desired formula for ${\mathcal N}$
in terms of $\chi_B$.
\end{proof}

Thus, we have defined the Euler characteristic of the BRST complex $C$ and its algebraic version $B$ and connect them with the degree of freedom. To define both of them, we uniformly assign large dimensions to the fields. 
Then, we have found that the  degree of freedom is equal to minus the residue at infinity of the BRST complex Euler characteristic logarithmic derivative.  At the same time, it is also equal 
to minus the residue at zero of the algebraic BRST complex Euler characteristic logarithmic derivative. 
For both versions of the BRST complex, the Euler characteristic is the only information needed to find the degree of freedom. 

\begin{exam}[DoF count using Euler characteristic in the Maxwell equations for the strength tensor] 
\label{ex: maxwell via Euler}
Let us construct BRST complex and calculate Euler characteristic for the Maxwell equations. We consider the same   Maxwell equations for the strength tensor as in Example~\ref{ex: maxwell}. 

In addition to the variables $F^{\mu 
\nu}$ for the six fields, we introduce 8 variables $A^*_\mu$
and $\widetilde A^*_\mu$ for the antifields (these antifields correspond in (\ref{Maxwell-eq}) to the equations for $T^\mu_1$ and 
$T^\mu_2$ respectively) and two variables $c^*$
and $\widetilde c^*$ for antighosts (these variables correspond to the gauge identities (\ref{I-Maxwell}) for $T_1$ and $T_2$). Since the complex ${\mathbf F}$ in (\ref{eq: F for Maxwell}) has only three terms and these terms are generated by the above variables, these variables also generate the whole BRST complex.

To define its Euler characteristic, let us assign some positive  differential order  $k_F$ to the field variables  $F^{\mu \nu}$. Then we have 
$$
\ord F^{\mu \nu} = k_F, 
\ord A^*_\mu = \ord \widetilde A^*_\mu=  k_F+1, 
\ord c^* = \ord \widetilde c^*=  k_F+2.
$$
The ghost numbers of these variables are 
$$
\gh F^{\mu \nu} = 0, 
\gh A^*_\mu = \gh \widetilde A^*_\mu=  1, 
\gh c^* = \gh \widetilde c^*=  2.
$$
So, the only nonzero Betti numbers of the complex 
${\mathbf F}$ defined in (\ref{eq: F for Maxwell}) are
$$
b_{0,k_A} = 6, b_{1,k_A+1} = 8,  
b_{2,k_A+2} = 2.
$$
Then we use the equations (\ref{eq: Poincare_for_BRST}) and (\ref{eq: Euler_for_BRST}) to calculate  the Poincare series and the Euler characteristic of the BRST complex as follows:
$$
P_C(t,z) = \frac{(1+tz^{k_F+1})^8}{(1-z^{k_F})^6 (1-t^2z^{k_F+2})^2}, \qquad
\chi_C(z) = P_C(-1,z)=\frac{(1-z^{k_F+1})^8}{(1-z^{k_F})^6 (1-z^{k_F+2})^2}.
$$
For the logarithmic derivative $l(z) = (\ln \chi_C(z))'$, we have
$$
l(z) = \left( \ln  \frac{(1-z^{k_F+1})^8}{(1-z^{k_F})^6 (1-z^{k_F+2})^2} \right) '
= -6\frac{-k_F z^{k_F-1}}{1-z^{k_F}} +
8\frac{-(k_F+1) z^k_F}{1-z^{k_F+1}}
-2\frac{-(k_F+2) z^{k_F+1}}{1-z^{k_F+2}}.
$$
Then the degree of freedom is equal to the limit 
$$
{\mathcal N} =\lim_{|z|\to \infty} z l(z) = 
-6 k_F +8 (k_F+1) -2 (k_F+2) = 4.
$$
We have obtained the same value ${\mathcal N} = 4$.
One can also note that this limit is equal to minus the residue, $ {\mathcal N} = -\Res(l(z), \infty) $.
\end{exam}

\section{Degree of freedom for non-homogeneous systems}

\label{sec: non-homogen}


Recall that the entries $\widehat T_{ai}$  of the matrix $\widehat T$ are polynomials from the ring $P$. 
Recall that we assign some integral degree $\theta_i$ to each unknown function $\phi^i(x)$ (where $i=1, \dots, m$) in system~(\ref{EoMs}). (By default, we put $\theta_i =0$.) 
Then the left-hand side of the $a$-th equation is the sum $\sum_i \widehat T_{ai} \phi^i(x)$. The {\em differential order}
$k_a$ of this equation denotes the maximal degree of a differential polynomial $\widehat T_{ai} \phi^i(x)$ appearing in this sum, that is, $k_a =\max_i (\deg \widehat T_{ai} + \theta_i)$. 

The highest-order homogeneous part of the system is called its  {\em symbol}. In details,  
by a symbol of the 
system~(\ref{EoMs}) we mean the system defined by the matrix $\lt \widehat T
$ ($\lt$ denotes `leading term') where each polynomial entry $\widehat T_{ai}$ is replaced 
its  part $\lt \widehat {T_{ai}} $ of the highest degree  $k_a-\theta_i$.  

The following is another interpretation of the matrix $\lt \widehat T
$. Consider $T$ as a linear map 
$T: F_1\to F_0$ of two free $P$-modules $F_1,F_0$. Assign grading to these free modules by putting $F_1 = \bigoplus_{i=1}^n P(-\theta_i)$
and $F_0 = \bigoplus_{a=1}^m P(-k_a)$. Then $\lt \widehat T$ is the homogeneous degree zero part of the map $T$.

\begin{defi}
	Let us call a matrix $T$ as above and the corresponding system~(\ref{EoMs}) {\em weakly involutive} if
	for each polynomial $p$ such that the equation 
	$p=0$ is a consequence of the system~(\ref{EoMs}),
	its higher homogeneous part $\lt p$ gives the equation $\lt p = 0$ which is a consequence of the symbol system defined by the matrix 
	$\lt \widehat T $.	 
\end{defi}

The main examples of weakly involutive systems are the following.

1. Homogeneous systems. By definition, the system is homogeneous if $\lt \widehat T = \widehat T$. Then the system is obviously weakly involutive.

2. Involutive systems (see~\cite{Seiler} for the definitions and discussion on the subject). If the system is involutive then it is weakly involutive for the degrees $\theta_i=0$ (this follows from \cite[Section 7.2]{Seiler}).
Still, there exist weakly involutive systems that are not involutive. For example, the system~$u_{xx}=0, u_{yy}=0$ is homogeneous (hence, involutive) but not  involutive~\cite[Example 7.2.2]{Seiler}. 

3. Suppose that the collection~$\{ T_a | a= 1, \dots , n\}$ of the system's left-hand parts forms a Groebner basis of some submodule in the free module $F_0 = P^m$
for some degree-compatible order. (For the definition and introduction to the theory of Groebner bases in polynomial modules, we refer the reader to~\cite{groebner} ).

Note that for each system of the form~(\ref{EoMs}), there exist both equivalent finite involutive system and a finite system such that its right-hand parts form a Groebner basis (say, for the degree-lexicographical order)~\cite{Seiler, groebner}. Either of the new systems is weakly involutive. Therefore, we obtain

\begin{theorem}
	\label{th: equivalent_involutive}
	For each system~(\ref{EoMs}) there exists an equivalent weakly involutive system.  
\end{theorem}


For an element $t$ of a free graded module $F$, let us denote here by $\lt t$ its highest homogeneous part. For a submodule $W\subset F$, by $\lt W$ we denote the homogeneous submodule generated by all elements $\lt t$ for $t\in F$.
Note that $\lt{\im \widehat T} \supset  \im \lt \widehat T$
and $\lt {\ker \widehat T} \subset \ker \lt \widehat T$.

\begin{prop}
	For a matrix $T$ and the fixed degrees $\theta_i$ as above, 
the following conditions are equivalent: 

(i) $T$ is weakly involutive;

(ii) $\lt {\im \widehat T} = \im \lt \widehat T$;

(iii) $\lt {\ker \widehat T} = \ker \lt \widehat T$.
\end{prop}

Note that if  $\widehat T: F_1 \to F_0$  is a map of two graded modules 
 $F_1 = \bigoplus_{j=1}^n P(-t_j)$
and $F_0 = \bigoplus_{i=1}^m P(-d_i)$, then $\widehat{T}^*: F_0^* \to F_1^*$ 
is naturally a  map of graded modules with dual grading,
$F_0^* = \bigoplus_{i=1}^m P(d_i)$ and $F_1^* = \bigoplus_{j=1}^n P(t_j)$.

\begin{defi}
\label{def: 2_weak_inv}
	Let us call the system~(\ref{EoMs}) and the matrix $\widehat T$ doubly weakly involutive if both $\widehat{T}$ and $\widehat{T}^*$ are weakly involutive with respect to a pair of mutually dual gradings and, in addition, $(\lt T)^* = \lt (T^*)$. 
\end{defi}

\begin{exam}
	\label{ex:hom_system_is_involutive}
	Suppose the system~(\ref{EoMs}) is homogeneous 
	with some choice of degrees $\deg \phi^i = \theta_i$. Then it is involutive by definition. Moreover, the conjugate system 
	in the conjugate variables ${\phi^*}^1, \dots , {\phi^*}^m$ 
	became homogeneous, if we put $\deg {\phi^*}^i = -\theta_i$.
	So, each homogeneous system is doubly weakly involutive. 
\end{exam}

\begin{exam}
	\label{ex: 1-equation_involutive}
	Suppose that the system~(\ref{EoMs}) consists of a single equation, $m=1$. 	
	Then $\widehat{T}$ is weakly involutive, while $\widehat{T}^*$ is generally not. Now, 
	let $G=(g_1, \dots, g_q)$ be a Groebner basis of the ideal $I$ generated by the entries of $\widehat{T}$ (for a degree-compatible ordering of monomials). If we replace the row matrix $\widehat{T} = (T_{11}, \dots , T_{1m})$ by the row $G$, 
	the system become doubly weakly involutive. The conjugate system defined by $G^*$ is equivalent to the system defined by $\widehat{T}^*$, so, the degrees of freedom does not change, $\nfree_{\widehat{T}} = \nfree_{G}$ and $\nfree_{\widehat{T}^*} = \nfree_{G^*}$.	Thus, if the system consists of a single equation, it is equivalent to a doubly weakly involutive one.  
\end{exam}

\begin{prop}
	Suppose that the system~(\ref{EoMs}) is double involutive. Then its degree of freedom $\nfree_{\widehat{T}}$ 
	is equal to the degree of freedom for the homogeneous system defined by the matrix $\lt \widehat{T}$. In particular, this degree of freedom can be calculated by Theorem~\ref{th:deg_free n graded 2-sided} applied to the two-sided complex ${\bf F}_{\lt \widehat{T}}$ of the form~(\ref{eq: 2-sided_resolution})  constructed by the matrix $\lt \widehat{T}$. 
\end{prop}

\begin{proof}
Consider the complexes ~(\ref{eq: our M resol}) and~(\ref{eq: resol for W}) constructed by the matrices $\lt \hat T$ and $(\lt T)^* = \lt T^*$ in place of $T$ and $T^*$. 
These complexes form graded resolutions ${\mathbf F}^{\gr V}$ and ${\mathbf F}^{\gr W}$ for the associated graded modules 
to $V$ and $W$. By~\cite[Theorem 1.8]{SR09} (this theorem is attributed to Robbiano, 1981; the proof in {\em Op.cit.}
does not use the locality assumption), 
there exist free resolutions ${\mathbf F}^{ V}$ and ${\mathbf F}^{W}$ for the modules $V$ and $W$ such that their associated graded complexes are the graded resolutions above, 
that is, $\gr {\mathbf F}^{ V} ={\mathbf F}^{\gr V} $ and $\gr {\mathbf F}^{W} = {\mathbf F}^{\gr W}$.
Then the associated graded complex $\gr {\mathbf F} $ to the two-sided complex ${\mathbf F}$ defined in~(\ref{eq: 2-sided_resolution}) is the analogous complex  $ {\bf F}_{\lt \widehat{T}} $ 
constructed by the graded resolutions ${\mathbf F}^{\gr V}$ and ${\mathbf F}^{\gr W}$,
$$
	\gr {\mathbf F} =  {\bf F}_{\lt \widehat{T}} :	0\to F_d^{\gr V}	\to\dots \to  F_1^{\gr V} \stackrel{\lt T^*}{\to} F_0^{\gr V} \stackrel{R^*}{\to} (F_2^{\gr W})^* \to \dots \to (F_d^{\gr W})^* \to 0.
$$

Recall that the only nonzero homologies of the complexes  
${\mathbf F}$ and $\gr {\mathbf F}$ may occur in the rightmost terms beginning with the term $F_0^{V}$ (respectively, $F_0^{(\gr\!) V}$). These homologies are 
$H_0({\mathbf F}) = \ext^1_P(W, P)$ (resp., $H_0(\gr {\mathbf F}) = \ext^1_P(\gr W, P)$) in the zero  term $F_0^{-}$ and  $H_{1-k}( {\mathbf F}) = \ext^k_P(W, P)$ (resp., $H_{1-k}(\gr {\mathbf F}) = \ext^k_P(\gr W, P)$) in terms $F_k^{(\gr\!) W})^*$, $k = 2, \dots, d$, where $(\gr\!)$ denotes $\gr$ for $\gr {\bf F}$ and the empty sign for${\mathbf F}$.  By  Proposition~\ref{prop: groth}, the modules $H_0(-)$ here have dimensions at most $d-1$,
while the dimensions of the modules $H_{1-k}(-)$ do not exceed $d-2$  for $k\ge 2$. It follows from Proposition~\ref{prop:dim_n_mult} that for large $N$,
$h_{H_{1-k} (\gr\!) {\mathbf F}} (N) = o(N^{d-2}) $ if $k\ge 2$
and 
$$
h_{H_{0} (\gr\!) {\mathbf F}} (N) = \frac{e}{(d-2)!}N^{d-2} + o(N^{d-2}),$$
where  
$$e = e(H_0((\gr\!) {\mathbf F}), d-1) = e(\ext^1_P((\gr\!) W, P), d-1).$$
By Theorem~\ref{th: deg_free_via_Ext}, this number $e$ is equal to the degree of freedom ${\mathcal N}_{\lt \widehat T} $
of the system defined by the matrix $\lt \widehat T$
for the complex $\gr {\bf F}$ and
to the degree of freedom 
${\mathcal N} = {\mathcal N}_{\widehat T} $
for the complex ${\bf F}$.

The Euler characteristic of the  $N$-th graded component of the complex ${\mathbf F}$ is equal to
$$
\chi_{\gr {\mathbf F}} (N) = \sum_{k=0}^d (-1)^k h_{  F_k^{\gr V}}(N) - \sum_{k=2}^d (-1)^k  h_{\left( F_k^{\gr W} \right)^*}(N)
= \sum_{k=-d}^d (-1)^k 
h_{   H_k (\gr {\mathbf F}) }(N). 
$$
So, the above  Euler characteristic is asymptotically equal to 
$$
\chi_{\gr {\mathbf F}} (N) = \frac{{\mathcal N}_{\lt \widehat T}}{(d-2)!}N^{d-2} + o(N^{d-2}).
$$

Now, we are able to apply Proposition~\ref{prop:Hilbert_pols_for_complexes} to the complex ${\mathbf F}$. In the notation of this proposition, we have 
$\delta  = d-1$. So, we deduce that 
$$
Q(N) :=  \sum_{k=0}^d (-1)^k h_{  F_k^{ V}}(N) - \sum_{k=2}^d (-1)^k  h_{(F_k^{ W})^*}(N) - \sum_{k=-d}^d (-1)^k h_{   H_k ( {\mathbf F}) }(N) 
$$ 
is a polynomial of degree at most $d-2$ for large $N$. 
Consider the first two summands here. Since 
$h_{  F_k^{ V}}(N) = h_{  \gr F_k^{ V}}(N)$ and 
$h_{( F_k^{ W})^*}(N) = h_{(  \gr F_k^{ W})^*}(N)$
by the definition of the Hilbert function, the sum of the first two summands is equal to 
$$
 \sum_{k=0}^d (-1)^k h_{  F_k^{ V}}(N) - \sum_{k=2}^d (-1)^k  h_{(F_k^{ W})^*}(N) =  \sum_{k=0}^d (-1)^k h_{  \gr F_k^{ V}}(N) - \sum_{k=2}^d (-1)^k  h_{(\gr F_k^{ W})^*}(N)
 = \chi_{\gr {\mathbf F}} (N)$$
 $$ =  \frac{{\mathcal N}_{\lt \widehat T}}{(d-2)!}N^{d-2} + o(N^{d-2}).
$$
At the same time, the third term here is equal to 
$$
-\sum_{k=-d}^d (-1)^k 
h_{   H_k ( {\mathbf F}) }(N) = 
-\left( h_{   H_0 ( {\mathbf F}) }(N) 
+ \sum_{k=-d}^{-1} (-1)^k
h_{   H_k ( {\mathbf F}) }(N) \right) 
= -\frac{{\mathcal N}_{ \widehat T}}{(d-2)!}N^{d-2} + o(N^{d-2}).
$$
Substituting the terms into the above formula for $Q(N)$, we get the asymptotic equality
$$
\frac{{\mathcal N}_{\lt \widehat T}}{(d-2)!}N^{d-2} 
- \frac{{\mathcal N}_{ \widehat T}}{(d-2)!}N^{d-2}+ o(N^{d-2}) = o(N^{d-2}).
$$
Thus, we deduce that ${\mathcal N}_{\lt \widehat T} = {\mathcal N}_{ \widehat T}$.
\end{proof}

\begin{exam}[Degree of freedom  count for massive spin 2 field]

\label{ex: massive spin 2}
Consider traceless rank two symmetric tensor field in Minkowski space
\begin{equation}\label{S2}
  S^{\mu\nu}=S^{\nu\mu}\,, \quad \eta_{\mu\nu}S^{\mu\nu}=0 \, .
  \end{equation}
 This field describe the irreducible massive spin two if it obey the system of equations
 \begin{equation}\label{MassiveSpin2}
  T^{\mu\nu}\equiv (\Box +m^2) S^{\mu\nu} =0,\, \quad T^\mu\equiv \partial_\nu S^{\mu\nu} = 0, \, \quad \Box=\eta^{\mu\nu}\partial_\mu\partial_\nu \, .
 \end{equation}
Put  $d=4$. Then there are  nine inhomogeneous equations  $T^{\mu\nu}$ of the second order, and four equations $T^\mu$ of the first order, $n_2=9,\, n_1=4$.
 There are four identities of there third order, $l_3=3$,
 \begin{equation}\label{IdS2}
   \partial_\mu T^{\mu\nu}- (\Box +m^2) T^\mu\equiv 0 \, .
 \end{equation}
 These identities are irreducible. The equations are not gauge invariant.

 The direct calculation (using Macaulay2) shows that  there are no further identities. Moreover, the matrix 
 $\lt \widehat T$ in this case is the analogous matrix $\widetilde T$ corresponding to the same system with $m=0$ (massless spin 2 field). Direct calculations show
 that the system is doubly weakly involutive in our case. Thus, we can apply the formula from Theorem~\ref{th:deg_free n graded 2-sided}.
Then the degree of freedom count reads,
\begin{equation}\label{DoFMassiveS2}
  N_{DoF}=9\cdot 2+ 4\cdot 1-4\cdot 3= 10 \, .
\end{equation}
 This is the correct degree of freedom number for the massive spin 2 field by the phase space count.
 \end{exam}

Let us call two systems of the form~(\ref{EoMs}) (or the corresponding matrices $\widehat{T}^0, \widehat{T}^1$) elementary equivalent if 
either the systems or their conjugates are equivalent; notation: $\widehat{T}^0\equiv \widehat{T}^1$.
 Moreover, we call two systems defined by matrices $\widehat{T}$ 
and $\widehat{T}'$ weakly equivalent if there is a collection of subsequently equivalent matrices $\widehat{T}=\widehat{T}^0 \equiv \widehat{T}^1\equiv \dots \equiv \widehat{T}^q = \widehat{T}'$.

 \begin{prop}
 	\label{prop: deg_free_invariant}
 	The degree of freedom 
 	$\nfree_C$ and the degree of freedom of the conjugate system 
 	$\nfree_{C'}$ are not changed if we replace the system with an equivalent one. 
 \end{prop}
 
 The proof follows from the next two lemmata. 
 \begin{lemma}
 	Suppose $\widehat T'$ is a matrix obtained from $\widehat T$
 	by adding zero columns and rows. Then $\nfree_{\widehat T'} = \nfree_{\widehat T}$.	 
 \end{lemma}
 
 \begin{proof}
 	If $\widehat T'$ has size $(m+a)\times (n+b)$, then we can construct a free resolution of 
 	the module ${W'} = \Coker \widehat T' = W\oplus P^a$ as follows:
 	$$	  F_{W'}:	\dots \to P^{r+b} \to P^{n+b} \stackrel{\widehat T'}{\to} {P}^{m+a}
 	\to {W'} \to 0. 
 	$$
 	It is a direct sum of the free resolution~(\ref{eq:right_resolution}) of $W$ and a complex 
 	$$
 	0\to 
 	P^b \stackrel{\Id}{\to} P^b 
 	\stackrel{0}{\to}
 	P^a \stackrel{\Id}{\to} P^a 
 	\to 0.
 	$$
 	Taking the homology of $\Hom_P( F_{W'}, P)$
 	at the term ${P}^{m+a}$, we get $\ext^1_P({W'}, P) = \ext^1_P({W}, P) = U$. So, 	$
 	\nfree_{\widehat T'} = e(U, d-1) = \nfree_{\widehat T}$.	
 \end{proof}

 \begin{lemma}
 	Suppose $\widehat T'$ is a matrix obtained from $\widehat T$
 	by adding to some row (or a column) a linear combination of other rows (respectively, columns). Then $\nfree_{\widehat T'} = \nfree_{\widehat T}$.	 
 \end{lemma}

 \begin{proof}
 	Both operations do not change the isomorphism class of the modules $W = \Coker \widehat T$ and $V = \Coker \widehat T^*$. So, the operations do not change 
 	the numbers $\nfree_{\widehat T} = e(\ext^1_P(W, P), d-1) $ and $\nfree_{\widehat T^*} = e(\ext^1_P(V, P), d-1) $. 
 \end{proof}

By 
Proposition~\ref{prop: deg_free_invariant}, we have

\begin{prop}
	If two systems with matrices $\widehat{T}$ 
and $\widehat{T}'$ are weakly equivalent, then $\nfree_{\widehat{T}} = \nfree_{\widehat{T}'}$ and $\nfree_{ {\widehat{T}}^*} = \nfree_{ {\widehat{T}'}{}^*}$.
\end{prop}

\begin{cor}
	Suppose that a system~(\ref{EoMs}) is weakly equivalent to a doubly weakly involutive system defined by a matrix $\widehat{T}'$. Then $\nfree_{\widehat{T}}$   can be calculated by the formula from Theorem~\ref{th:deg_free n graded 2-sided} applied to the two-sided complex constructed by the homogeneous matrix $\lt \widehat{T}'$. 
\end{cor}

In the view of Corollary~\ref{cor: conjugate N = N }, we deduce:

\begin{cor}
	Suppose that a system~(\ref{EoMs}) is weakly equivalent to a doubly weakly involutive system. 
 Then the degree of freedom for the system and its conjugate coincide, ${\mathcal N}^\dagger = {\mathcal N}$. 
\end{cor}

We have seen in Example~\ref{ex: 1-equation_involutive} 
that a system of a single equation is weakly equivalent to a doubly weakly involutive system. 

\begin{conj}
	Each system of the form~(\ref{EoMs}) is 
	weakly equivalent to a doubly weakly involutive system. 
\end{conj}

This conjecture implies

\begin{conj}
	Consider a system of the form~(\ref{EoMs}). Then the degree of freedom for the system and its conjugate coincide.
\end{conj}



\begin{exam}[\cite{KLS2013}]
\label{ex: proca}
    The (non-involutive an non-weakly-involutive) system of $d$ Proca equations 
	$$
	P_\mu\equiv (\delta^\mu_\nu \, \Box -\partial_\mu\partial^\nu - m^2 \, \delta_\mu^\nu)A_\nu = 0
	$$
	(where $m\ne 0$, $\Box = \partial^\nu \partial_\nu$, and $\mu, \nu$ run $ 0, \dots, d-1$)
	is equivalent to the  system  of $d$ Klein--Gordon equations and one additional equation (the transversality condition), 
	\begin{equation}
		\label{eq: KG}
	\begin{array}{lll}
		T_\mu & \equiv & (\Box -m^2) A_\mu = 0\\
		T_\bot & \equiv & \partial^\mu A_\mu=0 .
	\end{array}
    \end{equation}
Among these $d+1$ equations we have one gauge identity:
\begin{equation}
\label{Proca-Identity}
\partial^\mu T_\mu-T_\bot\equiv 0 \, .
\end{equation}
The above identity is of the third order because the first order differential operator $\partial^\mu$ acts on the left hand side of the second order Proca equations.
    
    A straightforward Groebner basis argument shows that the last system is doubly weakly involutive. The leading term  matrix $\lt \widehat{T}$ then corresponds  to this  homogeneous system 
	\begin{equation}
		\label{eq: KG homogeneous}
		\begin{array}{lll}
		\widehat T_\mu & \equiv & \Box A_\mu = 0\\
		T_\bot & \equiv & \partial^\mu A_\mu=0 
	\end{array}
	\end{equation}
 (that is, it corresponds to the case $m=0$). 
 The identity analogous to~(\ref{Proca-Identity}) also holds for the homogeneous equations. The free resolutions of the corresponding graded modules $V$ and $W$ are the following
 (see~\cite{KLS2013}): 
	$$
	0 \to P^d \stackrel{\lt \widehat T}{\to} P^{d}(2)\oplus P(1) \to W \to 0
	$$ 
	and $$
	0 \to P(-3) \to P^{d}(-2)\oplus P(-1) \stackrel{\lt \widehat {T}^*}{\to} P^{d} \to V \to 0. 
	$$
 The double complex~(\ref{eq: 2-sided_resolution}) then obtains the form 
 $$0 \to P(-3) \to P^{d}(-2)\oplus P(-1) \stackrel{\lt \widehat {T}^*}{\to} P^{d} \to 0 .
 $$
 Then the formula from Theorem~\ref{th:deg_free n graded 2-sided} gives
	$Q_F(z) = d-z-dz^2+z^3$ and 
	$$
	{\mathcal N} = -Q_F(1)' = -(-1-2dz+3z^2)|_{z=1} = 2d-2.
	$$

 Now, consider the conjugate system to~(\ref{eq: KG}). It consists of $d$ equations of the form
 $$
 {T^\dagger}_\mu  \equiv  (\Box -m^2) T_\mu + \partial^\mu T_\bot, \mu = 0, \dots, d-1.
 $$
 This system is doubly weakly involutive (because the system~(\ref{eq: KG}) is of this type). 
 If we assign the degrees $\deg T_\mu =-2 $ and 
 $\deg T_\bot = -1$ to the variables, the matrix $\lt {T^\dagger}$ obtains the form 
 $$
 \lt {T^\dagger}_\mu  \equiv  \Box T_\mu + \partial^\mu T_\bot, \mu = 0, \dots, d-1.
 $$

 Then the free resolutions of the graded modules $W$ and $V$ have the form
$$
	0 \to P^d \stackrel{\lt {\widehat{T^\dagger}}^*}{\to} P^{d}(2)\oplus P(1) \to V \to 0
	$$ 
	and $$
	0 \to P(-3) \to P^{d}(-2)\oplus P(-1) \stackrel{ \lt \widehat {T^\dagger}}{\to} P^{d} \to W \to 0. 
	$$
 The double complex~(\ref{eq: 2-sided_resolution}) looks as
 $$0 \to P^d \stackrel{\lt {\widehat{T^\dagger}}^*}{\to} P^{d}(2)\oplus P(1) \to P(3) \to 0
 $$
 (the term $P(3)$ appears here as the dual module to $P(-3)$).
Since the variables $T_\mu$ and $T_\bot$ are not of zero degree (in contrast to the variables $A_\mu$), we should apply Corollary~\ref{cor:main_homogeneous_general} in place of Theorem~\ref{th:main_homogeneous}. The formula gives 
$$
{\mathcal N} = -0\cdot d + (2d +1)-3 = 2d-2.
$$
The degree of freedom is the same for the two conjugate systems. 
\end{exam}

\vspace{0.2 cm}

\subsection*{Concluding remarks}

Let us make a few remarks about the possibilities of the further use of our results and the developed approaches.

For linear gauge systems, we used a two-sided complex of free polynomial moduli to establish a connection between the Einstein's DoF count recipe and the BRST complex. We believe that this complex, being the Fourier dual of BRST-complex, can be useful for studying the BRST cohomology of linear systems beyond the DoF number problem. In particular, it can provide a systematic tool for introducing auxiliary fields to bring non-Lagrangian field equations to an equivalent Lagrangian form.

We have established, for homogeneous system, an expression of the degree of freedom as minus the residue at infinity of the BRST complex Euler characteristic logarithmic derivative.
It seems interesting to express other physically meaningful
numerical and functional BRST invariants 
in terms of this Euler characteristics. 

We have found that DoF for the Hermitian conjugate systems is the same. This opens a new way to calculate this DoF for a gauge theory via the conjugate system and to construct the field theories which are equivalent at the linear level, but might not be equivalent with the inclusion of interactions. 

For topological field theories, the local DoF is zero. One can interpret this as follows: there is no free parameter of a general solution which is an arbitrary function of $d-1$ variables. However, there might be general functions of a smaller number of free variables which are parameters of the general solution. One can consider the numbers of such functions, say, of $k$ variables as a $k+1$--dimensional DoF of the theory. Our approach opens a way to study such smaller-dimensional degrees of freedom via the gauge structure of the theory and BRST cohomology.  

\subsubsection*{Acknowledgements}

The second author would like to express his gratitude for the hospitality of the University of Haifa, where part of this article was written. We thank A.A.~Sharapov for fruitful discussions.

The research described in Section~\ref{sec:gauge_and _polynom}
was carried out with the support of RSF grant No. 24-21-00341 (https://rscf.ru/project/24-21-00341/). 
The research described in Sections \ref{sec: algebra}, \ref{sec: non-homogen} was partially supported by the research project FSWM-2025-0007 of the Russian Ministry of Science and Higher Education.


\appendix

\section{Algebraic background: the Hilbert function, dimension,  and multiplicity of polynomial modules}

\label{sec: algebra}
\label{app: algebra}

In this appendix, we collect some basic definitions and theorems from commutative algebra. We focus on the Hilbert function, dimension, and multiplicity of polynomial modules. These concepts are essential for the algebraic interpretation of DoF. Note that while most textbooks and monographs on commutative algebra focus on the case of local rings, many important results also hold for modules over polynomial rings. We discuss only the versions relevant to polynomial modules.

We recall here the classical definitions and main properties of Hilbert functions, Hilbert series, Hilbert polynomials, dimension, and multiplicities of polynomial modules in Proposition~\ref{prop:dim_n_mult} and Corollary~\ref{cor: additive_multiplicity}. 
A more delicate property of Hilbert functions for complexes of filtered modules and their homologies~\cite{FCV99} is discussed in Proposition~\ref{prop:Hilbert_pols_for_complexes}. A key result about multiplicities of the Ext modules (the proof mimics Grothendieck's proof for local rings) in Proposition~\ref{prop: groth} is important for the connection of DoF with homologies. For a homogeneous module, we give a formula for the multiplicity as the value of the derivative of a certain polynomial (see Proposition~\ref{prop:Q and multiplicity}). This formula is used to calculate DoF.

As before, we denote by $P$ the ring of polynomials $P = \CC [\partial_0, \dots , \partial_{d-1}]$ of the variables $\partial_0, \dots , \partial_{d-1}$. 
Let $M$ be a finitely generated $P$-module. For details of the following classical results, we refer to~\cite{Serre}. Note that the same holds in the more complicated 
 case of algebraic D-modules (which correspond to differential operators with polynomial coefficients), see~\cite{borel, bjork}.  
 
 Denote by $P_N \subset P$ the linear span of all  
 monomials of degree $N$ ($N$ is an integer). 
 We call the module  $M$ graded if $M=
 \dots M_{-1 } \oplus M_0\oplus M_1 \oplus \dots$ with $P_N M_q \subset M_{N+q}$. The Hilbert function of a graded module is defined as $h(M,N) = \dim M_N$. Its generating function is the formal Laurent power series $H_M(z) = \sum_{N\ge 0} z^N \dim M_N$ which is called the Hilbert series of $M$.  
 
Now, consider a more general case of a non-graded module.
 A filtration $F$ on 
$M$ is a family of vector spaces $F_pM$ for integer $p$ (where $F_pM=0$ for $p<<0$) with the inclusions
$$
M = 
\dots \subset F_0M \subset F_1M \subset \dots 
\subset F_pM \subset \dots 
$$
We call a filtration admissible 
if all vector spaces $F_pM$ are finite-dimensional, the union $\bigcup_p F_pM$ is equal $M$, and $F_{p+q}M = P_{\le p} M_q$ for all large enough $q$ and all $p\ge 0$. For example, if $X $ is a finite generating set of $M$, then the filtration $F^X_pM = F_p P X$ is always admissible (where $F_p P = P_0\oplus \dots \oplus P_p$).  If $M'\subset M$ a submodule and $N = M/M'$ is a factor-module, then an admissible filtration $F_pM$ of $M$ induces an admissible filtration $F_pM' = M' \cap F_pM$ and $F_pN = F_p M / F_pM'$ on $N$.


If $M$ is graded as above, then $F_M = M_0\oplus \dots \oplus M_p $ is an admissible filtration on $M$. On the other hand, to each filtered module $M$ one can associate the graded module $\gr M = \oplus \gr M_N$ with $\gr M_N = F_N M / F_{N-1}M$ (where $F_{-1}M = 0$).  
Since the filtration is admissible, the module 
$\gr M$ is finitely generated. 
Then the Hilbert function (and the Hilbert series) of the filtered module $M$ are, by definition, the ones for $\gr M$: $h_M(N) = \dim \gr M_N$, $H_M(z) = H_{\gr M}(z)$.
In terms of the dimensions  $\tilde h_M(N) = \dim F_N M$, we get $h_M(N) = \tilde h_M(N) - \tilde h_M(N-1)$. 

The main properties of Hilbert functions and Hilbert series are the following, see~\cite[Ch.2]{Serre} 
 and~\cite[Sec. 1.3]{bjork}.  

\begin{prop}
	\label{prop:dim_n_mult}
(a) There exists a polynomial $P_M(N)$
(Hilbert--Samuel polynomial) such that $h_M(N) = P_M(N)$ for all large enough $N$. The degree $ \deg P_M(N)$ is equal to $D-1$, where $D = \dim M$ is the Krull dimension of $M$.  The leading coefficient of $P_M$ is equal to $e(M)/(D-1)!$, where the integer    $e(M) = P_M^{(D-1)}(0)$ is called the {\em multiplicity} of the module $M$. Both the dimension $d$ and multiplicity $e(M)$ do not depend on the choice of the admissible filtration on $M$.

(b) The Hilbert series of $M$ is a rational function of the form 
 $$
 H_M(z) = \frac{p_M(z)}{(1-z)^D}
 $$
 for some Laurent polynomial $p_M(z)$.
 Here the multiplicity of $M$ can be found as $e(M) = p_M(1)$. 
 
(c) The following additivity property holds: 
if 
$$0\to M' \to M \to M'' \to 0
$$ is an exact sequence of filtered modules (with $M$ as above), 
then $$P_M(N) = P_{M'}(N) +P_{M''}(N) + o(N^{D-1}).
$$
%
%
%
\end{prop}

The following more general version of  Proposition~\ref{prop:dim_n_mult}(c) have been appeared 
 in~\cite[Remark 2.4(1)]{FV93}
for modules over local rings. Another proof that appears in~\cite{FCV99} can easily be adapted for our case of modules over polynomial rings as well, see \cite[Proposition A.1.8 and Theorem 1.2.6]{FCV99}.

\begin{prop}
	\label{prop:Hilbert_pols_for_complexes}
Let
$$
M_\cdot: 0\to M_q \to \dots \to M_0\to 0
$$
be a complex of finite $P$-modules (not necessary exact)
and assume that, for some $\delta\ge 0$, all the homology modules of the associated graded complex $H_k(\gr M)$ have dimension at most $
\delta$. 
Denote $H_k = H_k(M)$. Then, the integer function
$$
Q(N) := \sum_k (-1)^k h_{M_k}(N) - 
       \sum_k (-1)^k h_{H_k}(N)
$$
is a polynomial of degree $<\delta$ for $N>>0$. 
\end{prop}

 \begin{rema}
 	{Another version of the Hilbert--Samuel polynomial is defined by the equality $\tilde h_M(N) = \widetilde P_M(N)$ for all large enough $N$. This means that $P_M(N) = \widetilde P_M(N) - \widetilde P_M(N-1)$.  Then $\deg \widetilde P_M(N)  = D$ and $e(M) = \widetilde P_M^{(D)}(0)$, so that 
 		$\tilde h_M(N) = N^D e(M)/D! + o(N^D)$.}
 \end{rema}


The last part of Proposition~\ref{prop:dim_n_mult} implies the following. For a module $M$ of dimension
at most $q$, put 
$$
e(M,q) = \left\{ \begin{array}{ll}
	e(M), & \dim M =q, \\
	0, & \dim M< q.
\end{array}
\right.
$$

\begin{cor}
\label{cor: additive_multiplicity}
	Suppose that the dimension of a finitely generated module $M$ 
	is not greater than $q$. If the sequence 
	$0\to M' \to M \to M'' \to 0
	$ of filtered modules is exact, then 
	$$e(M,q) = e(M',q)+e(M'',q).
	$$
\end{cor}

For the degree of freedom calculation, we are particularly interested in evaluating the multiplicity and the dimension for some $\ext$ modules. The following 
proposition is quite useful for bounding the dimension.  It is a variation of a part of the  Grothendieck local duality theorem (see the implication (b)$\Rightarrow$(a) in the proof of~\cite[Th. V.3.1]{sga2}).
 
 \begin{prop}
 	\label{prop: groth}
 	Let $X$ be a finitely generated $P$-module. Then $\dim \ext^i_P(X,P) \le d-i $ for each $i=1, \dots, d$. 
 \end{prop}

\begin{proof}
	If $p$ is  an arbitrary prime ideal of height $<i$, then 
	$\pd X_p \le \gd P_p = \dim P_p =  \height p <i$. 
	Since the localization functor is exact, it follows that $\ext^i_P(X,P)_p =
	 \ext^i_{P_p}(X_p,P_p) = 0$. 
  
  So, the support of  $\ext^i_P(X,P)$ does not contain any prime ideal of 
	 height $<i$. It follows that 
	 $\dim \ext^i_P(X,P) = \dim \supp \ext^i_P(X,P) 
	 \le \sup \{\coheight p | \height p \ge i\} \le d-i$.	
\end{proof}

For graded modules, 
 the Hilbert series and multiplicities have the following homological interpretation. Assume that a finitely generated $P$-module $M$ is graded. Then 
it admits a graded free resolution  
\begin{equation}
	\label{eq: free_resol_for_M}
{\bf F_M}:	\dots \to F_2^M \to F_1^M \to F_0^M, 
\end{equation}
where the last map has cokernel $M$ 
(so that the exact sequence can be continued up to $F_0\to M \to 0$) and $F_i^M = \bigoplus_{j}^P(-j)^{b_{ij}}  $
for some integer numbers $b_{ij} = b_{ij}^M$ (the Betti numbers of $M$). 
Note that here one can assume $F_i^M = 0$ for $i>d$ since the homological dimension of the ring $P$ is $d$. 
Since each free module $F_i^M$ is finitely generated, only a finite set of Betti numbers $b_{ij}$ is nonzero.
 
Let us denote the  Laurent polynomial 
$
\sum_j b_{ij}z^j$ by $Q^i_M(z)$. Then $H_{F_i^M}(z) = Q^i_M(z) H_P(z)$. 
If the resolution is minimal, then 
$
 Q^i_M(z) = H_{\tor_i^P(M,\CC)}(z).
$


Taking the Euler characteristic of the exact sequence~(\ref{eq: free_resol_for_M}), we get the equality of Hilbert series
$$
H_M(z) = \sum_i (-1)^i H_{F_i^M}(z),
$$
or 
$$
\frac{p_M(z)}{(1-z)^D} =  \frac{ \sum_{i} (-1)^i  Q^i_M(z)}{(1-z)^d}.
$$
Let us denote the last numerator $ \sum_{i} (-1)^i  Q^i_M(z) = \sum_{i,j} (-1)^i { b_{ij} z^j}$ by $Q_M(z)$. 
This is a Laurent polynomial with integer coefficients.  
It follows that 
\begin{equation}
\label{eq:Q(M)=(1-z)^(d-D)p_M}    
 Q_M(z) = (1-z)^{d} H_M(z) = (1-z)^{d-D} p_M(z).
\end{equation}

In contrast to the polynomials $p_M(z)$, the polynomials $Q_M(z)$ are additive.

\begin{cor}
	\label{cor:Q is additive}
If the sequence of homogeneous maps of graded modules 
$0\to M' \to M \to M'' \to 0
$  is exact, then 
$$
   Q_M(z) = Q_{M'}(z)+Q_{M''}(z).
$$
\end{cor}

\begin{proof}
This is a consequence of the Hilbert series equality 
 $H_M(z) = H_{M'}(z)+H_{M''}(z)$ and the formula $Q_M(z)  = (1-z)^{d} H_M(z)$  from (\ref{eq:Q(M)=(1-z)^(d-D)p_M}) accompanied with similar formulae for $Q_{M'}$ and $Q_{M''}$.
\end{proof}

By Proposition~\ref{prop:dim_n_mult}b, $e(M, d) = Q_M(1)$. A connection of the function $Q_M(z)$ with the multiplicity $e(M, d-1)$ is the following. 

\begin{prop}
	\label{prop:Q and multiplicity}
	If $\dim M \le d-1$ for a graded $P$-module $M$, then the multiplicity $e(M, d-1)$ is equal to minus the value of the derivative of $Q_M(z)$ at $z=1$, 
	$$
	e(M, d-1) = -Q_M'(1) = \sum_{i,j} j (-1)^{i+1} { b_{ij} }.
	$$
\end{prop}

\begin{proof}
	Suppose that  $ \dim M = d-1$. Then 
	$e(M,d-1) = e(M) = p_M(1)$. We have 
	$Q_M'(z) = \left[(1-z)p_M(z) \right]' = 
	(1-z) p_M'(z)-p_M(z) $, so that $Q_M'(1) = 
	-p_M(1) = - e(M,d-1)$. 
	
	Now, assume that $ D = \dim M < d-1$. Then 
	$e(M,d-1) = 0$. At the same time,  the polynomial $Q_M(z) = (1-z)^{d-D} p_M(z)$ is divisible by $(1-z)^2$, so that $Q_M'(1)=0$.
\end{proof}

\end{document}